\documentclass[11pt,a4paper]{article}
\pdfoutput=1
\usepackage{hyperref}
\usepackage{fullpage}
\usepackage{amssymb}
\usepackage{mathtools}
\usepackage{amsthm}
\usepackage{eucal}
\usepackage{dsfont}
\usepackage[T1]{fontenc}
\usepackage{lmodern}
\usepackage[stretch=10,shrink=10]{microtype}
\usepackage[capitalise]{cleveref}
\usepackage{color,soul}
\usepackage[linesnumbered,ruled,vlined]{algorithm2e}
\usepackage{graphicx}
\usepackage{overpic}
\usepackage{lipsum}
\usepackage[framemethod=default]{mdframed}

\mdfsetup{
    linewidth=0.7pt
}

\usepackage{amsmath}
\usepackage{amsfonts}
\usepackage{tikz-cd}
\allowdisplaybreaks[1]

\def\A{\mathcal{A}}

\def\E{\mathcal{E}}
\def\F{\mathcal{F}}

\def\H{\mathcal{H}}

\def\M{\mathcal{M}}
\def\N{\mathcal{N}}

\def\P{\mathcal{P}}

\def\S{\mathcal{S}}
\def\T{\mathcal{T}}

\def\X{\mathcal{X}}
\def\Y{\mathcal{Y}}
\def\Z{\mathcal{Z}}

\theoremstyle{plain}
\newtheorem{theorem}{Theorem}[section]
\newtheorem{lemma}[theorem]{Lemma}

\newtheorem{cor}[theorem]{Corollary}

\theoremstyle{definition}
\newtheorem{definition}[theorem]{Definition}
\newtheorem{claim}[theorem]{Claim}

\newtheorem{remark}[theorem]{Remark}
\newtheorem{fact}[theorem]{Fact}

\newcommand {\minusspace} {\: \! \!}
\newcommand {\smallspace} {\: \!}
\newcommand {\fn} [2] {\ensuremath{ #1 \minusspace \br{ #2 } }}
\newcommand {\Fn} [2] {\ensuremath{ #1 \minusspace \Br{ #2 } }}

\newcommand {\set} [1] {\ensuremath{ \left\lbrace #1 \right\rbrace }}

\newcommand {\br} [1] {\ensuremath{ \left( #1 \right) }}
\newcommand {\Br} [1] {\ensuremath{ \left[ #1 \right] }}

\newcommand {\norm} [1] {\ensuremath{ \left\| #1 \right\| }}
\newcommand {\normsub} [2] {\ensuremath{ \norm{#1}_{#2} }}
\newcommand {\onenorm} [1] {\normsub{#1}{1}}

\newcommand {\abs} [1] {\ensuremath{ \left| #1 \right| }}

\newcommand{\uninorm}[1]{{\left\vert\kern-0.25ex\left\vert\kern-0.25ex\left\vert #1
		\right\vert\kern-0.25ex\right\vert\kern-0.25ex\right\vert}}

\newcommand {\defeq} {\ensuremath{ \stackrel{\mathrm{def}}{=} }}
\newcommand {\mutinf} [2] {\fn{\mathrm{I}}{#1 \smallspace : \smallspace #2}}
\newcommand {\condmutinf} [3] {\mutinf{#1}{#2 \smallspace \middle\vert \smallspace #3}}
\newcommand {\relent} [2] {\fn{\mathrm{D}}{#1 \middle\| #2}}

\DeclareMathOperator*{\bigE}{\mathbb{E}}
\newcommand {\expec} [2] {\Fn{\bigE_{\substack{#1}}}{#2}}

\newcommand{\supp}[1]{\mathrm{supp}\br{#1}}

\newcommand {\id} {\ensuremath{\mathds{1}}}

\newcommand {\suppress}[1]{}

\newcommand {\mytitle} {One-Shot Slepian Wolf}

\newcommand{\misdiv}[3]{\fn{\mathrm{D}_s^{#1}}{#2 \middle\| #3}}
\newcommand{\htisdiv}[3]{\fn{\mathrm{D}_H^{#1}}{#2 \middle\| #3}}

\newcommand {\prob} [2] {\Fn{\Pr_{#1}}{#2}}

\hypersetup{
	pdfstartview={FitH},
	pdfdisplaydoctitle={true},
	breaklinks={true},
	bookmarksopen={true},
	bookmarksnumbered={false},
	pdftitle={\mytitle},
}

\begin{document}
\title{On the compression of messages in the multi-party setting}

\author{
Anurag Anshu\footnote{Centre for Quantum Technologies, National University of Singapore, 117543, Singapore. \texttt{a0109169@u.nus.edu}} \qquad
Penghui Yao\footnote{State Key Laboratory for Novel Software Technology, Nanjing University, Nanjing, 210093, PR China. \texttt{pyao@nju.edu.cn}} 
}

\maketitle

\begin{abstract}
We consider the following communication task in the multi-party setting, which involves a joint random variable $XYZMN$ with the property that $M$ is independent of $YZN$ conditioned on $X$ and $N$ is independent of $XZM$ conditioned on $Y$. Three parties Alice, Bob and Charlie, respectively, observe samples $x,y$ and $z$ from $XYZ$. Alice and Bob communicate messages to Charlie with the goal that Charlie can output a sample from $MN$ having correct correlation with $XYZ$. This task reflects the {\em simultaneous message passing} model of communication complexity. Furthermore, it is a generalization of some well studied problems in information theory, such as distributed source coding, source coding with a helper and one sender and one receiver message compression. It is also closely related to the lossy distributed source coding task.

Our main result is an achievable communication region for this task in the one-shot setting, through which we obtain a near optimal characterization using auxiliary random variables of bounded size. We employ our achievability result to provide a near-optimal one-shot communication region for the task of lossy distributed source coding, in terms of auxiliary random variables of bounded size. Finally, we show that interaction is necessary to achieve the optimal expected communication cost for our main task.
  
\end{abstract}

\section{Introduction}

{\em Source coding} is a central task in information theory, where the task for a sender is to communicate a sample from a source. The constraint is that the error made by the receiver in the decoding process should be small and the goal is to communicate as less number of bits as possible. A tight characterization of this task was achieved by Shannon~\cite{Shannon} in the asymptotic and i.i.d. setting, where the senders are assumed to have a large number of identical and independent samples from the source. Later, Slepian and Wolf~\cite{SlepianW73} presented a tight characterization for a multi-party source coding in the same asymptotic and i.i.d. setting. The powerful techniques introduced by these authors were further generalized to asymptotic and non-i.i.d. setting~\cite{Han03}. 

In the recent years, there has been a growing interest in the study of various generalizations of source coding in the non-asymptotic and non-i.i.d. setting. An important setting that has been actively investigated in the past few decades is the {\em one-shot setting}, where just one sample from the source is given to the senders. A notable generalization of source coding, that has been studied in both the asymptotic i.i.d. and the one-shot settings, is that the sender observes a sample $x$ from a source and the receiver is supposed to output a random variable that depends on the sample. This task was investigated in the one-shot setting in~\cite{Jain:2003, HJMR10} in the context of {\em communication complexity} (known there as message compression), while in the asymptotic and i.i.d setting, it was studied in~\cite{BennettDHSW14, Cuff13} in the context of {\em channel simulation}. The task can be stated in more details as follows, where Alice is the sender and Charlie is the receiver.

\vspace{0.1in}

\noindent {\bf Task A:} Let $XM$ be joint random variables taking values over a set $\X\times \M$. Alice and Charlie possess pre-shared randomness, which is independent $X$. Alice observes a sample $x$ from $X$ and communicates a message to Charlie , where the message depends on the input $x$ and the value observed from the pre-shared randomness. Charlie outputs a sample distributed according to a random variable $M'$ satisfying $\frac{1}{2}\|XM-XM'\|_1 \leq \epsilon.$ 

\vspace{0.1in}

Above, $\|.\|_1$ is the $\ell_1$ distance and $\epsilon$ is an error parameter.  It was shown in~\cite{HJMR10} that the \textit{expected communication cost} of this task is equal to $\mutinf{X}{M}$ (up to a additive factor) in the one-shot setting, generalizing the result of Huffman~\cite{Huffman52}. The work~\cite{HJMR10} also gave important applications to the {\em direct sum} results in two-party communication complexity. The work~\cite{BravermanRao11} considered an extension of Task A with side information about $X$ at Charlie.

\vspace{0.1in}

\noindent {\bf Task B:} Let $XMZ$ be joint random variables taking values over a set $\X\times \M\times \Z$, such that $M-X-Z$. Alice and Charlie possess pre-shared randomness independent of $XZ$. Alice observes a sample $x$ from $X$ and Charlie observes a sample $z$ from $Z$. Alice communicates a message to Charlie, where the message only depends on the input $x$ and the pre-shared randomness. Charlie outputs a sample distributed according to a random variable $M'$ such that $\frac{1}{2}\|XMZ-XM'Z\|_1 \leq \epsilon.$ 

\vspace{0.1in}

An important assumption in Task B is the Markov chain condition $M-X-Z$, which signifies the fact that $M$ is to be treated as a `message' generated by Alice given $x$. Essentially the same condition arises when side information $Z$ is available with receiver in the context of channel simulation, as the channel generates $M$ only depending on $X$. The authors in~\cite{BravermanRao11} obtained the expected communication cost of $\condmutinf{X}{M}{Z}$ (the conditional mutual information) up to additive factors, in the one-shot setting. Recently, it has been shown that the protocol in~\cite{BravermanRao11} is near optimal also in terms of the worst communication cost~\cite{AnshuJW17un}.  

In this work, we consider a generalization of Task B in the setting of two senders and one receiver. More precisely, we consider the following task, which was also studied in~\cite{AnshuJW17un}.

\vspace{0.1in}

\noindent {\bf Task C:} Let $XYZMN$ be joint random variables taking values over a set $\X\times \Y \times \Z\times\M\times \N$, and satisfying the Markov chain conditions $M-X-YZN$ and $MXZ-Y-N$. Alice and Charlie possess pre-shared randomness and Bob and Charlie independently possess another pre-shared randomness. Alice observes a sample $x$ from $X$, Bob observes a sample $y$ from $Y$ and Charlie observes a sample $z$ from $Z$. Alice and Bob respectively communicate a message to Charlie (which also depends on the value observed from the pre-shared randomness). Charlie outputs a sample distributed according to a random variable $M'N'$ such that $\frac{1}{2}\|XYZMN-XYZM'N'\|_1 \leq \epsilon.$

\vspace{0.1in}

Task C is a generalization of the distributed source coding (DSC) studied by Slepian and Wolf~\cite{SlepianW73} in the asymptotic and i.i.d setting and in \cite{uteymatsu-matsuta-2014, TanKosut14, Warsi2016, AnshuJW17un} in the second order and one-shot settings. We show in Appendix~\ref{append:SCHtaskB} that Task C also generalizes the task of source coding with a helper (SCH), which has been studied in~\cite{Wyner75, miyakaye-kanaya-1995, WatanabeKT15, Verdu12, Warsi2016, uteymatsu-matsuta-2015, AnshuJW17un}. The motivation for considering Task C is for the message compression in multi-party communication complexity. In the past two decades, many elegant message compression protocols in the one-shot setting have been discovered in the context of communication complexity~\cite{Jain:2003, HJMR10,BravermanRao11,BarakBCR13} (some of which we discussed earlier). These protocols show how to achieve the communication cost close to the {\em information complexity}~\cite{Braverman15}, which measures the amount of information exchanged between the communicating parties. As a result, significant progress has been made towards the direct sum problems, one of the central open problems in communication complexity.  However, the notion of information complexity in the multi-party communication complexity has not yet been established,  party due to the fact that the communication cost region for multi-party communication is more involved and less understood. Hence, giving a tight characterization of the communication cost of Task C is a first step towards developing a correct notion of information complexity in the multi-party communication complexity.

We begin by trying to understand the rate region for Task C in the asymptotic and i.i.d. setting. By employing the time sharing technique \cite[Section 4.4]{GamalK12}, it can be found that the following is an achievable rate region, where $R_1$ is the rate of communication from Alice to Charlie and $R_2$ is the rate of communication from Bob to Charlie.  
\begin{eqnarray}
\label{eq:timesharetaskB}
R_1&\geq&\condmutinf{X}{M}{NZ} \nonumber\\
R_2&\geq&\condmutinf{Y}{N}{MZ}\nonumber\\
R_1+R_2&\geq&\condmutinf{XY}{MN}{Z},
\end{eqnarray}
where on the right hand side, we have the mutual information quantities. Is it possible to show that this rate region is optimal? The answer is negative, which can be seen by considering the task of SCH, which is a special case of Task C as discussed earlier. In this task, Alice holds a random variable $X$ and Bob holds a random variable $Y$ correlated with $X$. Alice and Bob communicate messages to Charlie in a manner that Charlie is able to output $X$ with high probability. 

It is well known~\cite[Section 10.4]{GamalK12} that the time sharing rate region for the SCH task (as obtained by setting $M=X$ and $N$ trivial in Eq.~\eqref{eq:timesharetaskB}) is not the optimal rate region. In fact, the known characterization of an optimal rate region requires the introduction of auxiliary random variables. Thus, the rate region given in Eq.~\eqref{eq:timesharetaskB} is not an optimal characterization of Task C and an optimal characterization may require some auxiliary random variables. On the other hand, the utility of the achievable rate region in Eq.~\eqref{eq:timesharetaskB} is that it only involves the random variables that are input for the task.

\subsection*{Our results}

We obtain the following results in the paper.
\begin{itemize}
\item We study Task C in the one-shot setting. First, we show how to obtain a one-shot analogue of Eq.~\eqref{eq:timesharetaskB} in Theorem~\ref{thm:main}, which is our main result. Observe that time sharing method cannot be applied in the one-shot setting. Hence we require a new tool to obtain our result, which we achieve by an appropriate multi-partite generalization of the protocols constructed in~\cite{BravermanRao11, AnshuJW17un}. While this achievable rate region is not known to be optimal (which is not known even in the asymptotic and i.i.d. setting, as discussed earlier) it has the following applications. 
\begin{enumerate}
\item We obtain a nearly tight characterization of the {\em one-shot lossy distributed source coding}  (in presence of side information at the receiver) in Theorem~\ref{thm:lossyopt}. 
\item We obtain a nearly tight characterization of Task C with auxiliary random variables (of bounded size that is comparable to the size of input random variables) in the one-shot setting, in Theorem \ref{thm:auxchar}. 
\item In Section \ref{subsec:recoverDSC} we recover the near optimal one-shot results on the DSC task and one-sender one-receiver message compression task as obtained in~\cite{AnshuJW17un}.   
\end{enumerate}

\item We study the expected communication cost of {\em Selpian-Wolf task}, a special case of Task C,  where $Y,N$ are trivial and $M=X$. This two-party task (as Bob is not involved) was considered by Slepian and Wolf~\cite{SlepianW73} and the communication rate in the asymptotic and i.i.d. setting was shown to be equal to $\mathrm{H}(X|Z)$. We show in Theorem~\ref{theo:expeclowbound} that any one-way protocol for this task will incur an expected communication cost of $\frac{1}{\sqrt{\epsilon}}\mathrm{H}(X|Z)$. On the other hand, the result of~\cite{BravermanRao11} implies that there is an interactive protocol achieving the expected communication cost of $\mathrm{H}(X|Z) + c\br{\sqrt{\mathrm{H}(X|Z)} +\log\frac{1}{\epsilon}}$, for some universal constant $c$. Thus, there is a stark contrast between the one-way protocols and interactive protocols in terms of the expected communication cost.

In turn, this implies that one-way protocols cannot achieve the region given in Eq.~\eqref{eq:timesharetaskB} for Task C in expected communication, even when the side information $Z$ is trivial. On the other hand, we also observe that there is a simple interactive scheme that uses the protocol of~\cite{BravermanRao11} for Task B as a subroutine and achieves the region given in Eq.~\eqref{eq:timesharetaskB} (up to small additive factors).

\end{itemize}

\subsection*{Our techniques}

For our achievability result in Theorem \ref{thm:main}, we use the two tools of convex-split \cite{AnshuDJ14} and position-based decoding \cite{AnshuJW17}. Convex-split technique allows the encoder to find the appropriate correlations in a collection of independent random variables (by incurring a small error) and position-based decoding is a \textit{hypothesis testing} process applied to a collection of random variables. Hypothesis testing is a technique to distinguish a random variable $X$ from a random variable $X'$ (both taking values over the same set $\X$), by constructing a test that accepts $X'$ with as small probability as possible, with the constraint that the same test must accept $X$ with probability close to $1$. Typically, it suffices to assume that the test corresponds to checking the membership of a sample in some suitable subset $S\subset\X$. Owing to the protocol developed in \cite{AnshuJW17un} for Task B, our key technical challenge will be to construct an appropriate hypothesis testing step, which we discuss in details in Subsection \ref{subsec:proofidea}. We remark that Task C was also studied in~\cite[Theorem 4]{AnshuJW17un} using the techniques of convex-split and position-based decoding. But the hypothesis testing step therein was different and in Subsection~\ref{subsec:comp} we show that our communication region contains the communication region obtained in~\cite[Theorem 4]{AnshuJW17un}, up to small additive factors.

For the proof of Theorem \ref{theo:expeclowbound}, we construct a joint random variable $XZ$ with two properties. First is that it has a low conditional entropy $\mathrm{H}(X|Z)$ and there is a value $z_0$ of $Z$ such that $(X|Z=z_0)$ has high entropy. The second property is that any one-way protocol with small expected communication cost for the Slepian-Wolf task leads to a protocol with small expected communication cost for the source coding of random variable $(X|Z=z_0)$. Since $(X|Z=z_0)$ has high entropy, it requires high expected length for its source coding \cite{Huffman52}. This leads to a contradiction.

\subsection*{Organisation}

We discuss our notations and the facts required in our proofs in Section \ref{prelims}. In Section \ref{sec:mainsec}, we discuss our achievability result and its various consequences. We also provide an overview of earlier techniques that are relevant to us. Various consequences of our results are discussed in Section \ref{sec:cons} In Section \ref{sec:expeccomm}, we discuss the expected communication cost of Task C.

\section{Preliminaries}
\label{prelims}

For a natural number $n$, let $[n]$ denote the set $\{1, 2, \ldots, n\}$. Let the random variable $X$ take values in a finite set $\X$ (all sets we consider in this paper are finite).  We let $p_X$ represent the distribution of $X$, that is for each $x \in \X$,  $p_X(x)\defeq\prob{}{X=x}$. We use $x\sim X$ to represent that $x$ is sampled from $X$. The support of the random variable $X$ is defined to be $\{x: p_X(x)>0\}$ and is denoted by $\supp{X}$. For any subset $\A\subseteq\X$, we use $\prob{X}{\A}$ to represent the probability $\sum_{x\in\A}p_X\br{x}$, that is, the probability that $x\in \A$.  Let random variables $XY$ take values in the set $\X\times\Y$. The random variable $Y$ conditioned on $X=x$ is denoted as $(Y|X=x)$. We say that $X$ and $Y$ are independent and denote the joint distribution by $X\times Y$,  if for each $x\in\X$ and $y\in\Y$ it holds that $p_{XY}\br{x,y}= p_X\br{x}p_Y\br{y}$.  We say that random variables $(X,Y,Z)$ form a Markov chain, represented as  $Y-X-Z$, if for each $x \in \X$, $(Y|X=x)$ and $(Z|X=x)$ are independent. For an event $E$, its complement is denoted by $\neg E$. The indicator random variable is denoted by the symbol $\id(\cdot)$. For a random variable $X$ taking values over $\X$ and a function $f: \X \rightarrow \X'$, we denote by $f(X)$ the random variable obtained by sampling $x$ according to $X$ and then applying $f$ to it. We use the same notation when $X$ is correlated with other random variables. For a random variable $X$ over a set $\X$ and a set $S \subseteq \X$, the random variable $X'$ defined as $p_{X'}(x) = \frac{\id(x\in S)p_X(x)}{\Pr_{X}\{G\}}$ is called a \textit{restriction} of $X$ over the set $S$.
	
\begin{definition}\label{def:infoquantitites}
Given $\epsilon\in(0,1)$ and random variables $X$ and $X'$ taking values in $\X$, we define
\begin{itemize}
\item $\ell_1$ distance. $$\onenorm{X-X'}\defeq\sum_{x\in\X}\abs{p_X\br{x}-p_{X'}\br{x}}.$$

\item KL-divergence. $$\relent{X}{X'}\defeq\sum_{x\in\X}p_X\br{x}\log\frac{p_X\br{x}}{p_{X'}\br{x}},$$ where we assume $0\log\frac{0}{0}=0$. If $\supp{X}\not\subseteq\supp{X'}$, then $\relent{X}{X'}\defeq\infty.$

\item Max information spectrum divergence. $$\misdiv{\epsilon}{X}{X'}\defeq\min\set{a:\prob{x\sim X}{\frac{p_X\br{x}}{p_{X'}\br{x}}\geq 2^a}\leq\epsilon}.$$

\item Hypothesis testing information spectrum divergence. $$\htisdiv{\epsilon}{X}{X'}\defeq\max\set{a:\prob{x\sim X}{\frac{p_X\br{x}}{p_{X'}\br{x}}\geq 2^a}\geq1-\epsilon}.$$
\end{itemize}
\end{definition}

\begin{fact}\label{fac:distancerestriction}
Let $G\subseteq\X$, $X$ be a random variable over $\X$ and $X'$ be the restriction of $X$ over the set $G$. It holds that 
\[\frac{1}{2}\onenorm{X-X'}=1-\prob{X}{G}.\]
\end{fact}
\begin{proof}
Consider
\[\onenorm{X-X'}=\sum_{x\in G}\abs{p_X\br{x}-\frac{p_X\br{x}}{\prob{X}{G}}}+\sum_{x\notin G}p_X\br{x}=\frac{\prob{X}{G^c}}{\prob{X}{G}}\cdot \prob{X}{G}+\prob{X}{G^c}=2\prob{X}{G^c}.\]
\end{proof}

\begin{fact}
Let  $XYZ$ and $X'Y'Z'$ be two joint distributions over $\X\times\Y\times\Z$. It holds that 
\[\relent{XYZ}{X'Y'Z'}\geq\relent{XY}{X'Y'}.\]
\end{fact}

\begin{fact}
Let $X$ and $Y$ be two distributions over set $\X$. It holds that 
\[\onenorm{X-Y}\leq2\cdot\sqrt{\relent{X}{Y}}.\]
\end{fact}

\begin{fact}\label{lem:distancemonotone}
Let $X$  and $X'$ be two random variable distributed over set $\X$ and $f:\X\rightarrow\Z$ be a map. It holds that 
\[\onenorm{f\br{X}-f\br{X'}}\leq\onenorm{X-X'}.\]
\end{fact}

The following is the classical version of the convex-split lemma~\cite{AnshuDJ14}, as stated in~\cite{AnshuJW17un}.

\begin{fact}[Convex-split lemma~\cite{AnshuJW17un}]
\label{fact:convsplit}
Let $\epsilon,\delta>0$, $R$ be a non-negative integer, $XM$ be a joint distribution over $\X\times\M$ and $W$ be a random variable distributed over $\M$
\[R\geq\misdiv{\epsilon}{XM}{X\times W}+2\log\frac{3}{\delta}.\]
Let $J$  be a random variable uniformly distributed over $[2^R]$ and the joint distribution $JXM_1M_2\ldots M_{2^R}$ be defined to be:
\[\prob{}{XM_1\ldots M_{2^R}=xm_1\ldots m_{2^R}|J=j}\defeq p_{XM}\br{xm_j}p_W\br{m_1}\cdots p_W\br{m_{j-1}}p_W\br{m_{j+1}}\cdots p_W\br{m_{2^R}}.\]
Then 
\[\frac{1}{2}\onenorm{XM_1M_2\ldots M_{2^R}-X\times W\times W\times\ldots\times W}\leq \epsilon+\delta.\]
\end{fact}

Its bipartite generalization is as follows, given in~\cite{AnshuJW17un}.
\begin{fact}[Bipartite convex-split lemma~\cite{AnshuJW17un}]
\label{genconvexcomb}
Let $\epsilon, \delta\in (0,1)$. Let $XMN$ (jointly distributed over $\X\times \M\times \N$), $U$ (distributed over $\M$) and $V$ (distributed over $\N$) be random variables. Let $R_1, R_2$ be natural numbers such that,
$$\prob{x,m,n\sim XMN}{\frac{p_{XM}(x,m)}{p_X(x)p_U(m)} \leq \frac{\delta^2}{24}\cdot 2^{R_1} \mbox{ and } \frac{p_{XN}(x,n)}{p_X(x)p_V(n)} \leq \frac{\delta^2}{24}\cdot 2^{R_2} \atop\mbox{ and }\quad\frac{p_{XMN}(x,m,n)}{p_X(x)p_U(m)p_V(n)} \leq \frac{\delta^2}{24}\cdot 2^{R_1+R_2}} \geq 1-\epsilon.$$
Let $J$ be uniformly distributed in $[2^{R_1}]$, $K$ be independent of $J$ and be uniformly distributed in $[2^{R_2}]$ and  joint random variables $\br{J,K,X,M_1,\ldots, M_{2^{R_1}},N_1,\ldots, N_{2^{R_2}}}$ be distributed as follows:
\begin{align*}
&\prob{}{(X,M_1, \ldots, M_{2^{R_1}}, N_1, \ldots, N_{2^{R_2}}) = (x, m_1, \ldots, m_{2^{R_1}}, n_1, \ldots , n_{2^{R_2}})| J=j, K=k}\\
&=  p_{XMN}(x,m_j,n_k)\cdot p_{U}(m_1)\cdots p_{U}(m_{j-1}) \cdot p_{U}(m_{j+1})\cdots p_{U}(m_{2^{R_1}}) \cdot \\ 
&\hspace{35mm} p_V(n_1)\cdots p_{V}(n_{k-1}) \cdot p_{V}(n_{k+1})\cdots p_{V}(n_{2^{R_2}}).
\end{align*}
Then (below for each $j \in [2^{R_1}], p_{U_j} = p_U$ and for each $k \in [2^{R_2}], p_{V_k} = p_V$),
 $$ \frac{1}{2}\|XM_1 \ldots M_{2^{R_1}}N_1 \ldots N_{2^{R_2}} - X\times U_1 \times\ldots \times U_{2^{R_1}}\times V_1 \times\ldots \times V_{2^{R_2}}\|_1 \leq \epsilon + \delta.$$
\end{fact}

We will also need a classical version of position-based decoding~\cite{AnshuJW17}, which was obtained in \cite{AnshuJW17un}. Below we provide a more rigorous proof. 

\begin{lemma}\label{lem:error}
Given $\epsilon\in[0,1)$ and joint distribution $CH$ over $[c+1]\times \H$, where $\H$ is a finite set and $\supp{C}=[c]$, we define $H_i\defeq\br{H|C=i}$. For any $i\in[c]$, let $\A_i\subset\H$ be a subset satisfying that $\prob{H_i}{\A_i}\geq1-\epsilon$.  Consider a protocol $\P$ which takes a sample $h\sim H$ and sequentially verifies whether $h\in\A_i$ for $i\in[c]$. $\P$ terminates and outputs the first $i$ satisfying $h\in\A_i$. Otherwise, it outputs $c+1$. The output is denoted by the random variable $C'$. It holds that 
\[\frac{1}{2}\onenorm{HC-HC'}\leq \sum_ip_C\br{i}\sum_{j:j\neq i,j\leq c}\prob{H_i}{\A_j}+ \epsilon.\]
\end{lemma}

\begin{proof}

From the definition, we have 
\begin{eqnarray}
&&\onenorm{HC'-HC}=\onenorm{\sum_ip_C\br{i}\br{HC'|C=i}-\sum_ip_C\br{i}H_i\times\id\br{C'=i}}\nonumber\\
&&\leq\sum_ip_C\br{i}\onenorm{\br{HC'|C=i}-H_i\times\id\br{C'=i}},\label{eqn:hc}
\end{eqnarray}
where the equality follows from the fact that $H_i\times\id\br{C=i}=H_i\times\id\br{C'=i}$. We will bound each of term in the summation on the right hand side. 

For any $i$, set $B_i\defeq\id\br{C'< i}$ for all $i\in[c+1]$. By Fact~\ref{lem:distancemonotone},
\begin{eqnarray}
&&\onenorm{\br{HC'|C=i}-\br{H|C=i}\times\id\br{C'=i}}\nonumber\\
&&=\onenorm{\br{HC'B_i|C=i}-\br{H|C=i}\times\id\br{C'=i}\times\id\br{B_i=0}}.\label{eqn:hcb}
\end{eqnarray}
Thus it suffices to upper bound the right hand side. For this, we introduce two sub-protocols $\P^0_i$ and $\P^1_i$. $\P^0_i$  runs $\P$ till step $i-1$ and sets $B_i=1$ and $C'$ to be the index returned upon $\P$'s termination, if $\P$ terminates. It sets $B_i=0, C'=i$ otherwise. If $B_i=0$, $\P^1_i$ runs $\P$ from step $i$ onwards and outputs $HC'B_i$. If $B_i=1$, $\P^1_i$ outputs $HC'B_i$ and then terminates. Observe that the output of $\P$ is part of the output of $\P^1_i\circ\P^0_i$. Consider 
\begin{eqnarray}
&&\onenorm{\br{HC'B_i|C=i}-\br{H|C=i}\times\id\br{C'=i}\times\id\br{B_i=0}}\nonumber\\
&&=\onenorm{\P^1_i\circ\P^0_i\br{H_i}-H_i\times\id\br{C'=i}\times\id\br{B_i=0}}\nonumber\\
&&\leq\onenorm{\P^1_i\circ\P^0_i\br{H_i}-\P^1_i\br{H_i\times\id\br{C'=i}\times\id\br{B_i=0}}}\nonumber\\
&&+\onenorm{\P^1_i\br{H_i\times\id\br{C'=i}\times\id\br{B_i=0}}-H_i\times\id\br{C'=i}\times\id\br{B_i=0}}\nonumber\\
&&\leq\onenorm{\P^0_i\br{H_i}-H_i\times\id\br{C'=i}\times\id\br{B_i=0}}\nonumber\\
&&+\onenorm{\P^1_i\br{H_i\times\id\br{C'=i}\times\id\br{B_i=0}}-H_i\times\id\br{C'=i}\times\id\br{B_i=0}}.\label{eqn:hcb4}
\end{eqnarray}
where the first inequality follows from Lemma~\ref{lem:distancemonotone}. Let $\E_{<i}$ be the event that the protocol $\P$ does not terminate in any step before $i$. Then using the expansion
\[H_i=\prob{H_i}{\E_{<i}}\cdot\br{H_i|\E_{<i}}+\prob{H_i}{\neg\E_{<i}}\cdot\br{H_i|\neg \E_{<i}}\]
we have 
\begin{eqnarray}
&&\P^0_i\br{H_i}=\prob{H_i}{\E_{<i}}\cdot\P^0_i\br{\br{H_i|\E_{<i}}}+\prob{H_i}{\neg\E_{<i}}\cdot\P^0_i\br{\br{H_i|\neg\E_{<i}}}\nonumber\\
&&=\prob{H_i}{\E_{<i}}\cdot\br{\br{H_i|\E_{<i}}\times\id\br{C'=i}\times\id\br{B_i=0}}+\prob{H_i}{\neg\E_{<i}}\cdot\P^0_i\br{\br{H_i|\neg\E_{<i}}},  \label{eqn:p0hi}
\end{eqnarray}
where the last equality follows from the fact that $\P^0_i$  outputs $B_i=0$ and $C'=i$ when it runs on $\br{H_i|\E_{<i}}$ and does not terminate till step $i-1$. Moreover,
\begin{eqnarray}
&&\br{H_i\times\id\br{C'=i}\times\id\br{B_i=0}}\nonumber\\
&&=\prob{H_i}{\E_{<i}}\cdot\br{\br{H_i|\E_{<i}}\times\id\br{C'=i}\times\id\br{B_i=0}}\nonumber\\
&&+\prob{H_i}{\neg\E_{<i}}\cdot\br{\br{H_i|\neg\E_{<i}}\times\id\br{C'=i}\times\id\br{B_i=0}}. \label{eqn:hicb}
\end{eqnarray}
Combining Eqs.~\eqref{eqn:p0hi}\eqref{eqn:hicb}, we have 
\begin{eqnarray}
&&\onenorm{\P^0_i\br{H_i}-H_i\times\id\br{C'=i}\times\id\br{B_i=0}}\nonumber\\
&&=\onenorm{\prob{H_i}{\neg\E_{<i}}\cdot\P_i^0\br{\br{H_i|\neg\E_{<i}}}-\prob{H_i}{\neg\E_{<i}}\cdot\br{H_i|\neg\E_{<i}}\times\id\br{C'=i}\times\id\br{B_i=0}}\nonumber\\
&&=\prob{H_i}{\neg\E_{<i}}\cdot\onenorm{\P_i^0\br{\br{H_i|\neg\E_{<i}}}-\br{H_i|\neg\E_{<i}}\times\id\br{C'=i}\times\id\br{B_i=0}}\nonumber\\
&&\leq 2\prob{H_i}{\neg\E_{<i}}\leq2\sum_{j<i}\prob{H_i}{\A_j}.\label{eqn:hcb3}
\end{eqnarray}
To upper bound the second term in Eq.~\eqref{eqn:hcb4}, consider
\begin{eqnarray}
&&\P^1_i\br{H_i\times\id\br{C'=i}\times\br{B_i=0}}\nonumber\\
&&=\prob{H_i}{\A_i}\cdot\P^1_i\br{\br{H_i|\A_i}\times\id\br{C'=i}\times\br{B_i=0}}+\prob{H_i}{\neg\A_i}\cdot\P^1_i\br{\br{H_i|\neg\A_i}\times\id\br{C'=i}\times\br{B_i=0}}\nonumber\\
&&=\prob{H_i}{\A_i}\br{H_i|\A_i}\times\id\br{C'=i}\times\br{B_i=0}+\prob{H_i}{\neg\A_i}\P^1_i\br{\br{H_i|\neg\A_i}\times\id\br{C'=i}\times\br{B_i=0}}, \label{eqn:p1i}
\end{eqnarray}
where the second equality follows from the definitions of $\P^1_i$ and $\br{H_i|\A_i}$.

Combining Eqs.~\eqref{eqn:hicb}\eqref{eqn:p1i}, we have 
\begin{eqnarray}
&&\onenorm{\P^1_i\br{H_i\times\id\br{C'=i}\times\id\br{B_i=0}}-H_i\times\id\br{C'=i}\times\id\br{B_i=0}}\nonumber\\
&=&\onenorm{\prob{H_i}{\neg\A_i}\cdot\P^1_i\br{\br{H_i|\neg\A_i}\times\id\br{C'=i}\times\id\br{B_i=0}}-\atop \prob{H_i}{\neg\br{\A_i}}\cdot\br{H_i|\neg\A_i}\times\id\br{C'=i}\times\id\br{B_i=0}}\nonumber\\
&=&\prob{H_i}{\neg\A_i}\cdot\onenorm{\P^1_i\br{\br{H_i|\neg\A_i}\times\id\br{C'=i}\times\id\br{B_i=0}}-\atop\br{H_i|\neg\A_i}\times\id\br{C'=i}\times\id\br{B_i=0}}\nonumber\\
&\leq&2\prob{H_i}{\neg\A_i}\leq 2\epsilon. \label{eqn:hcb5}
\end{eqnarray}

Combining Eqs.~\eqref{eqn:hc}\eqref{eqn:hcb}\eqref{eqn:hcb4}\eqref{eqn:hcb3}\eqref{eqn:hcb5}, we conclude the result.

\end{proof}

\section{Achievability result for two senders and one receiver}
\label{sec:mainsec}

We revisit our main task (Task C), restated here for convenience.

\begin{definition}
\label{def:task1}
\noindent {\bf A $\br{R_1, R_2, \epsilon}$ two-senders-one-receiver message compression with side information at the receiver:} Given joint random variables $XYZMN$ satisfying that $M-X-(Y,Z)$ and $N-Y-(X,Z)$, there are three parties Alice, Bob and Charlie holding $X$, $Y$ and $Z$, respectively. Alice sends a message of $R_1$ bits to Charlie and Bob sends a message of $R_2$ bits to Charlie. Charlie outputs a sample distributed according to the random variable $M'N'$ such that $\frac{1}{2}\onenorm{XYZMN-XYZM'N'}\leq\epsilon$. Shared randomness is  allowed between Alice and Charlie and between Bob and Charlie. 
\end{definition}

Before proceeding to our main result, we discuss the protocols in~\cite{BravermanRao11, AnshuJW17un} for the case of one sender and one receiver (Task B), and discuss the limitation of the known techniques for the two sender and one receiver case.

\subsection{Revisiting previous protocols for Task B}
\label{subsec:proofidea}

The idea in~\cite{BravermanRao11} is as follows, which we rephrase in the context of worst case communication. Alice and Charlie share sufficiently large number of copies of a random variable which is uniform over the set $\M\times [K]$, where $K$ is a sufficiently large integer. The number of required copies turns out to be approximately $|\M|$ and they index these copies with a unique integer in $[|\M|]$.  Given $x$, Alice has the knowledge of the probability distribution $p_{M|X=x}$ and given $z$, Charlie has the knowledge of the probability distribution $p_{M |Z=z}$. Define $$2^c\defeq \max_{m, x,z : p_{XZ}(x,z)>0}\frac{p_{M|X=x}(m)}{p_{M|Z=z}(m)},$$ as the largest possible ratio between these probabilities, for all $x,z$ in the support of $p_{XZ}$. It can be viewed as a one-shot analogue of the conditional mutual information, $\condmutinf{X}{M}{Z}$. Using the rejection sampling method of~\cite{vonNeumann51, Jain:2003, HJMR10}, Alice finds an index of the shared randomness where the sample $(m,e)$ satisfies $e\leq Kp_{M|X=x}(m)$. Charlie accepts an index if the associated sample $(m,e)$ satisfying $e\leq K\cdot 2^c\cdot p_{M|Z=z}(m)$. The definition of $c$ ensures that Charlie definitely accepts the index that Alice accepts (among various other indices). Then Alice uses hash functions to inform Charlie about the correct index (see~\cite{BravermanRao11} for more details).

In~\cite{AnshuJW17un}, it was shown that above protocol can be viewed in the context of hypothesis testing over an `extended distribution' over $\M\times [K]$. Using the convex-split~\cite{AnshuDJ14} and position-based decoding~\cite{AnshuJW17} methods, Charlie's operation is replaced by a hypothesis testing operation. Alice and Charlie divide their copies of shared randomness in approximately $2^c$ blocks. Alice uses convex-split method to find an index property correlated with $X$ (alternatively, she could have used rejection sampling). At this index, any sample $(m,e)$ satisfies $e\leq Kp_{M|X=x}(m)$. Denote the random variable associated to $e$ as $E$, which extends $XM$ to $XME$. Alice communicates the block number to Charlie using approximately $c$ bits of communication. Within this block lies the correct index that Alice wants Charlie to pick up. Charlie needs to distinguish this index with other indices, where any sample $(m,e)$ has the property that $e$ is uniform in $[K]$. Charlie uses hypothesis testing for this step, with the observation that the definition of $c$ ensures that $e \leq K\cdot 2^c\cdot p_{M|Z=z}(m)$ at the correct index, giving it a small support size. Thus, the test  that Charlie uses for hypothesis testing is simply to check the membership of $(m,e)$ in the support of the random variable $ME$. It accepts the distribution $ME$ at the correct index with probability $1$ and accepts the uniform distribution over $\M\times [K]$ with probability at most $\frac{2^{c}}{|\M|}$ (see~\cite{AnshuJW17un} for details).

One can try the same approach for Task C, by extending the random variables $M$ and $N$ into new random variables $ME$ and $NF$, where $e$ is uniform in $[K\cdot p_{M|X=x}(m)]$ and $f$ is uniform in $[K\cdot p_{N|Y=y}]$ for a given $x,y,m,n$. Alice and Bob can perform the convex-split steps and Charlie can attempt to perform the hypothesis testing step. The key challenge is to construct the correct test for hypothesis testing. Below, we analyze the test given in \cite{AnshuJW17un} by considering the full support of the distribution $MENF$. For the ease of argument, we assume that $Z$ is trivial.

The test should ensure that Charlie accepts the uniform distribution over $\M\times [K]\times \N\times [K]$ with probability at most $2^{-c'}$, where $c'$ is define as follows: 
$$\prob{x,y,m,n\sim XYMN}{\frac{p_{M|X=x}(m)p_{N|Y=y}(n)}{p_{MN}(m,n)} \leq 2^{c'}} \geq 1-\epsilon.$$
Observe that $c'$ is a one-shot analogue of the mutual information $\mutinf{XY}{MN}$ (see Equation \ref{eq:timesharetaskB}).
By the construction of the random variable $EF$, the definition of $c'$ can only ensure that $e\cdot f \leq K^2\cdot 2^{c'}p_{MN}(m,n)$. If $2^{c'}p_{MN}(m,n) \geq 1$ then all pairs $e,f$ satisfy the condition. Let the set of all such $(m,n)$ be `$\mathrm{Bad}$' and let the rest be `$\mathrm{Good}$'. For $(m,n) \in \mathrm{Good}$, the number of pairs $(e,f)$ that satisfy this condition can be calculated to be $$K^2\cdot 2^{c'}p_{MN}(m,n)\log\frac{1}{2^{c'}p_{MN}(m,n)}.$$ The support of $(EF\mid MN=m,n)$, for all $(m,n)\in \mathrm{Bad}$,  accepts the uniform distribution over $[K]\times [K]$ with probability $1$ and hence gives no advantage for hypothesis testing. Furthermore, the probability of the set `$\mathrm{Bad}$' can be large. To argue, consider
$$0\leq \sum_{m,n\in\mathrm{Good}}p_{MN}(m,n)\log\frac{1}{2^{c'}p_{MN}(m,n)} \leq \mathrm{H}(MN) - c'\cdot\prob{MN}{\mathrm{Good}},$$
which implies 
$$\prob{MN}{\mathrm{Good}} \leq \frac{\mathrm{H}(MN)}{c'}.$$
We have the following claim.
\begin{claim}
\label{clm:counterexample}
Fix an $\epsilon\in(0,1)$ and let $\alpha\in(0,1)$ be such that $\alpha(1-\alpha)^2 = 2\epsilon$. There exists a random variable $XYMN$ with $\M=\X=\Y=\N$ such that $\mathrm{H}(MN) \leq (1-\alpha^3)\log|\X| + 3$ and $c' \geq \log|\X| - 3$. Thus, 
$$\frac{\mathrm{H}(MN)}{c'} \leq 7\sqrt{\epsilon}.$$
\end{claim}
The proof of this claim is given in Appendix~\ref{append:counterex}. This implies that the probability  
$\prob{MN}{\mathrm{Bad}}$ is at least $1-7\sqrt{\epsilon}$, leading to large error. Hence, constructing hypothesis test using the full support of the random variable $MNEF$ does not give the desired result.

To solve the problem, we revisit the protocol in \cite{AnshuJW17un} for Task B and obtain a new hypothesis test. We recall the joint random variables $XMEZ$, where $XMZ$ satisfy $M-X-Z$ and $E$ is uniform in $[K\cdot p_{M|X=x}(m)]$ conditioned on $m,x$. The distribution of $E$, conditioned on $m,z$ and averaged over $x$, satisfies
\begin{eqnarray*}
p_{E|MZ=m,z}(e) &=& \sum_{x: e \leq K\cdot p_{M\mid X=x}(m)} \frac{1}{K\cdot p_{M\mid X=x}(m)}\cdot p_{X| MZ=mz}(x) \\ &\geq& \frac{1}{2^c\cdot K\cdot p_{M\mid Z=z}(m)}\sum_{x: e \leq K\cdot p_{M\mid X=x}(m)} p_{X| MZ=m,z}(x),
\end{eqnarray*}
where we have used the definition of $c$. Let $G$ be a random variable, jointly correlated with $MZ$, that takes the value $K\cdot p_{M|X=x}(m)$ with probability $p_{X|MZ=mz}(x)$ (by perturbing the conditional distribution in a negligible manner and choosing large enough $K$, we can assume that $K\cdot p_{M|X=x}(m)$ is unique for every $x,m$). Using this, above inequality simplifies to 
\begin{equation}
\label{eq:probElowbound}
p_{E|MZ=m,z}(e) \geq \frac{1}{2^c\cdot K\cdot p_{M\mid Z=z}(m)}\sum_{g: e \leq g} p_{G| MZ=m,z}(g) = \frac{\prob{g\sim G| MZ=m,z}{g\geq e}}{2^c\cdot K\cdot p_{M\mid Z=z}(m)}.
\end{equation}
We now show the following result.

\begin{lemma}\label{lem:1}
Let $EG$ be joint random variables taking values over $[K]\times[K]$ satisfying that $p_{E|G=g}\br{e}=0$ if $e>g$. Then for any $\delta\in(0,1)$, it holds that 
\[\prob{e\sim E}{\prob{g\sim G}{g\geq e} \leq\delta}\leq\delta.\]
\end{lemma}
\begin{proof}
Let $g^*$ be the smallest integer such that $\prob{g\sim G}{g\geq g^*}\leq\delta$. Then
\begin{eqnarray*}
\prob{e\sim E}{\prob{g\sim G}{g\geq e}}&=&\sum_{e\geq g^*}p_E\br{g}\\
&=&\sum_{e\geq g^*}\sum_{g\geq e}p_{EG}\br{e,g}\\
&=&\sum_{g\geq g^*}\sum_{e: g^*\leq e\leq g}p_{EG}\br{e,g}\\
&\leq&\sum_{g\geq g^*}p_G\br{g}\leq\delta.
\end{eqnarray*}
\end{proof}
\noindent Let $\A_{m,z}$ be the set of all $e$ for which
$$p_{E|MZ=m,z}(e) \geq \frac{\delta}{2^c\cdot K\cdot p_{M\mid Z=z}(m)}.$$
From Eq.~\eqref{eq:probElowbound} and Lemma~\ref{lem:1}, we find that $$\prob{E|MZ=m,z}{\A_{m,z}} \geq 1-\delta.$$
Moreover, $$|\A_{m,z}| \leq \frac{2^c\cdot K\cdot p_{M\mid Z=z}(m)}{\delta}.$$
Hence, the `test' $\A_{m,z}$ can be used to distinguish the random variable $(E|MZ=m,z)$ from the uniformly distributed random variable over $[K]$. The probability of accepting the uniform distribution over $[K]$ is at most $2^c\cdot\frac{p_{M\mid Z=z}(m)}{\delta}$. Thus the probability of accepting the uniform distribution over $\M \times [K]$ is at most
$$\frac{2^c}{|\M|}\cdot\sum_m\frac{p_{M\mid Z=z}(m)}{\delta} = \frac{2^c}{|\M|\delta}$$
for any $z$. This reproduces the property of the hypothesis test obtained in~\cite{AnshuJW17un} up to a small multiplicative factor of $\frac{1}{\delta}$. Moreover, this construction generalizes to the multi-variate setting, as shown in Lemma~\ref{lem:setA}.

\subsection{Achievable rate region for Task C}

Following is our main theorem.

\begin{theorem}\label{thm:main}
Given $\epsilon, \delta\in(0,1)$ such that $\frac{1}{\sqrt{\delta}}$ is an integer, let $R_1,R_2$ satisfy
\begin{equation}\label{eqn:main}
\prob{x,y,m,n,z\sim XYZMN}{\frac{p_{M|X=x}\br{m}}{p_{M|N=n, Z=z}\br{m}}\leq\delta\cdot 2^{R_1}~\mbox{and}~\frac{p_{N|Y=y}\br{n}}{p_{N|M=m, Z=z}\br{n}}\leq\delta\cdot 2^{R_2}\atop\mbox{and}~\frac{p_{M|X=x}\br{m}p_{N|Y=y}\br{n}}{p_{MN\mid Z=z}\br{mn}}\leq\delta^4\cdot2^{R_1+R_2-\log\log\frac{\max\set{\abs{\M},\abs{\N}}}{\delta}}}\geq1-\epsilon.
\end{equation}
There exists a $\br{R_1+3\log\frac{1}{\delta}, R_2+3\log\frac{1}{\delta}, \epsilon + 8\delta}$ two-senders-one-receiver message compression with side information at the receiver. 
\end{theorem}

\begin{proof}
We assume that $p_{M|X=x}\br{m}$ and $p_{N|Y=y}\br{n}$ are distinct rationals for all $\br{x,y,m,n}$, which is possible by perturbing the distributions and introducing an arbitrary small error.
Let $K$ be a sufficient large integer such that $Kp_{M|X=x}\br{m}$ and $Kp_{N|Y=y}\br{n}$ are integers for all $x,y,m,n$. We define random variables $EF$ over $[K]\times[K]$ such that $E$ is generated conditioned on $MX$ and $F$ is generated conditioned on $NY$ as follows: 
\[p_{E|MX=mx}\br{e}\defeq\frac{\id\br{e\leq Kp_{M|X=x}\br{m}}}{Kp_{M|X=x}\br{m}}~\mbox{and}~p_{F|NY=ny}\br{f}\defeq\frac{\id\br{f\leq Kp_{N|Y=y}\br{n}}}{Kp_{N|Y=y}\br{n}}.\]
We further define random variables 
\begin{equation}\label{eqn:ST}
S\defeq \tilde{S}\times L ~\mbox{and}~T\defeq \tilde{T}\times L
\end{equation}
where $\tilde{S}, \tilde{T}$ and $L$ are uniformly distributed over $\M, \N$ and $[K]$, respectively. Set 
\begin{equation}\label{eqn:1}
r_1\defeq\log\frac{\abs{\M}}{2^{R_1}}~\mbox{and}~ r_2\defeq\log\frac{\abs{\N}}{2^{R_2}}, R_3\defeq \lceil R_1+2\log 3/\delta\rceil, R_4\defeq\lceil R_2+2\log 3/\delta\rceil.
\end{equation}

Let $J_1$  be uniformly distributed over $[2^{R_3+r_1}]$ and the joint random variable $J_1XS'_1S'_2\ldots S'_{2^{R_3+r_1}}$ be defined to be:
\begin{eqnarray*}
&&\prob{}{X,S'_1\ldots S'_{2^{R_3+r_1}}=x,(m_1,e_1)\ldots (m_{2^{R_3+r_1}}, e_{2^{R_3+r_1}})|J_1=j}\defeq\\ && p_{XME}\br{x,m_j,e_j}p_S\br{m_1,e_1}\cdots p_S\br{m_{j-1},e_{j-1}}p_S\br{m_{j+1}, e_{j+1}}\cdots p_S\br{m_{2^{R_3+r_1}}, e_{2^{R_3+r_1}}}.
\end{eqnarray*}
Similarly, let $J_2$  be uniformly distributed over $[2^{R_4+r_2}]$ and the joint random variable $J_2YT'_1T'_2\ldots T'_{2^{R_4+r_2}}$ be defined to be:
\begin{eqnarray*}
&&\prob{}{Y,T'_1\ldots T'_{2^{R_4+r_2}}=y,(n_1,f_1)\ldots (n_{2^{R_4+r_2}}, f_{2^{R_4+r_2}})|J_2=j}\defeq\\ && p_{YNF}\br{y,n_j,f_j}p_T\br{n_1,f_1}\cdots p_T\br{n_{j-1},f_{j-1}}p_T\br{n_{j+1}, f_{j+1}}\cdots p_T\br{n_{2^{R_3+r_1}}, f_{2^{R_4+r_2}}}.
\end{eqnarray*}
We further assume that $XS'_1S'_2\ldots S'_{2^{R_3+r_1}}J_1$ and $YT'_1T'_2\ldots T'_{2^{R_4+r_2}}J_2$ are independent. From the choice of $R_3+r_1$ and $R_4+r_2$ and Fact~\ref{fact:convsplit}, we conclude that
\begin{eqnarray*}
&&\frac{1}{2}\|XS'_1\ldots S'_{2^{R_3+r_1}}- X\times S\times \ldots \times S\|_1 \leq \delta \\
&& \frac{1}{2}\|YT'_1\ldots T'_{2^{R_4+r_2}}- Y\times T\times \ldots \times T\|_1 \leq \delta.
\end{eqnarray*}
Since the random variables $S'_1\ldots S'_{2^{R_3+r_1}}$ are defined conditioned on $x$, and similarly the random variables $T'_1\ldots T'_{2^{R_4+r_2}}$ are defined conditioned on $y$, we conclude that
\begin{eqnarray}
&&\frac{1}{2}\onenorm{XYZS'_1\ldots S'_{2^{R_3+r_1}}T'_1\ldots T'_{2^{R_3+r_1}} - XYZ\times S\times \ldots S\times T\times \ldots \times T}\nonumber\\
&=&\frac{1}{2}\sum_{xyz}p_{XYZ}\br{x,y,z}\onenorm{\br{S'_1\ldots S'_{2^{R_3+r_1}}|X=x}\times\br{T'_1\ldots T'_{2^{R_3+r_1}}|Y=y}-S\times \ldots S\times T\times \ldots \times T}\nonumber\\
&\leq& 2\delta.\label{convclose}
\end{eqnarray}
We are now ready to define the protocol.

\begin{mdframed}
\textbf{Input:} Random variables $XYMN$ distributed over $\X\times\Y\times\M\times\N$, where $\X,\Y,\M$ and $\N$ are finite sets; reals $R_1,R_2,\epsilon,\delta$ satisfying Theorem~\ref{thm:main}; $R_3,R_4$ as defined in Eq.~\eqref{eqn:1}. Alice, Bob and Charlie are given $x, y$ and $z$, respectively, where $\br{x,y,z}\sim XYZ$.
\bigskip


\noindent \textbf{Shared resources:} Alice and Charlie share $S_1\ldots S_{2^{R_3+r_1}}$, which are $2^{R_3+r_1}$ copies of i.i.d. samples of $S$. Bob and Charlie share $T_1\ldots T_{2^{R_4+r_2}}$, which are $2^{R_4+r_2}$ copies of i.i.d. samples of $T$. Here $S$ and $T$ are defined in Eq.~\eqref{eqn:ST}.

\bigskip

\noindent \textbf{The protocol:}
\begin{enumerate}

\item Alice observes a sample $(x, (m_1, e_1), \ldots (m_{2^{R_3+r_1}}, e_{2^{R_3+r_1}}))$ from $XS_1\ldots S_{2^{R_3+r_1}}$ and samples $j_1$ from the conditional distribution $$\br{J_1|X,S'_1,\ldots, S'_{2^{R_3+r_1}}=x,(m_1, e_1),\ldots,(m_{2^{R_3+r_1}}, e_{2^{R_3+r_1}})}.$$

\item Bob  observes a sample $(y, (n_1, f_1), \ldots (n_{2^{R_4+r_2}}, f_{2^{R_4+r_2}}))$ from $YT_1\ldots T_{2^{R_4+r_2}}$ and samples $j_2$ from the conditional distribution $$\br{J_2|Y,T'_1,\ldots, T'_{2^{R_4+r_2}}=y,(n_1, f_1),\ldots,(n_{2^{R_4+r_2}}, f_{2^{R_4+r_2}})}.$$

\item Alice sends $j'_1\defeq\lceil\frac{j_1}{2^{r_1}}\rceil$ to Charlie.

\item Bob sends $j_2'\defeq\lceil\frac{j_2}{2^{r_2}}\rceil$ to Charlie.

\item Charlie selects the first pair 

$$\br{j_1,j_2}\in\set{j_1'\cdot 2^{r_1}+1,\ldots, j_1'\cdot 2^{r_1}+2^{r_1}}\times\set{j_2'\cdot 2^{r_2}+1,\ldots, j_2'\cdot 2^{r_2}+2^{r_2}}$$ 

in the lexicographical order such that $\br{m_{j_1}, e_{j_1},n_{j_2}, f_{j_2},z}\in\A$, where $\A$ is obtained from Lemma~\ref{lem:setA}. Charlie outputs $\br{m_{j_1},n_{j_2}}$. If no such index exists , Charlie outputs an arbitrary pair of elements in $\M\times\N$. Let the output of Charlie be the random variable $M'N'$.

\end{enumerate}

\end{mdframed}

\begin{remark}
		By virtue of Eq.~\eqref{convclose}, the global joint distribution of $XYZS_1\ldots S_{2^{R_3+r_1}}T_1\ldots T_{2^{R_4+r_2}}$ is close to that of $XYZS'_1\ldots S'_{2^{R_3+r_1}}T'_1\ldots T'_{2^{R_3+r_1}}J_1J_2$. 
	
\end{remark}

From Eq.~\eqref{eqn:1} the communication cost between Alice and Charlie is $R_1+3\log\frac{1}{\delta}$ and the communication cost between Bob and Charlie is $R_2+3\log\frac{1}{\delta}$. Moreover, applying Lemma~\ref{lem:setA} (to get a guarantee on the error on each step of Charlie's decoding), Lemma~\ref{lem:error} (to bound the error of Charlie's decoding) and Eq.~\eqref{convclose} (to bound the error due to the encoding of Alice and Bob), we obtain

\[\frac{1}{2}\onenorm{XYZMN-XYZM'N'}\leq \br{\epsilon+5\delta+2\delta+2^{r_1}\frac{\delta2^{R_1}}{\abs{\M}}+2^{r_2}\frac{\delta2^{R_2}}{\abs{\N}}+2^{r_1+r_2}\frac{\delta 2^{R_1+R_2}}{\abs{\M}\abs{\N}}}\leq \epsilon+10\delta.\]

\end{proof}

Lemma~\ref{lem:setA}, to be shown below, gives error bounds on the hypothesis testing part employed by Charlie in the proof of Theorem~\ref{thm:main}. It can be viewed as the multi-variate generalization of Lemma~\ref{lem:1}. 

\begin{lemma}\label{lem:setA}
Let $\delta\in(0,1)$ satisfy the condition that $\frac{1}{\sqrt{\delta}}$ is an integer greater than $1$. Then there exists a set $\A\subseteq\M\times\N\times\E\times\F\times \Z$ such that 
\[\prob{MNEFZ}{\A}\geq1-\epsilon-5\delta;\]
\[\prob{MEZ\times T}{\A}\leq\frac{\delta2^{R_2}}{\abs{\N}}, \prob{S\times NFZ}{\A}\leq\frac{\delta2^{R_1}}{\abs{\M}},\]
and
\[\prob{S\times T\times Z}{\A}\leq\frac{\delta 2^{R_1+R_2}}{\abs{\M}\abs{\N}}.\]
\end{lemma}

\noindent {\it Proof outline:} Lemma~\ref{lem:setA} is a two dimensional extension of Lemma~\ref{lem:1}. For any $x, y$, set $w_m\br{x}\defeq K\cdot p_{M|X=x}\br{m}$ and $v_n\br{y}\defeq K \cdot p_{N|Y=y}\br{n}$, and let $W,V$ be the corresponding random variables jointly distributed with $XYZMNEF$.   Fix the values $m,n,z$ and consider the joint distribution of $WVEF$.  We wish to show that $EF$ is `mostly' supported on a set of small size, using an argument similar to Lemma~\ref{lem:1}. This amounts to showing that the set of $(e,f)$, for which $w>e$ and $v>f$ with small probability according to $WV$, has small probability according to $EF$. Unfortunately, it is not clear how to achieve it for arbitrary $WVEF$ just with the condition that $E<W$ (always) and $F<V$ (always). But we know that $E \mid (W=w)$ is uniform in $[w]$ and $F\mid (V=v)$ is uniform in $[v]$, and we can use this property in our argument. 

As an illuminating example, we consider the case where $WV$ is supported in $\{a^0, \ldots a^1\}\times \{b^0, \ldots b^1\}$ (Figure~\ref{fig:WVsupport}), and it holds that $\frac{(a^1-a^0)(b^1-b^0)}{a^1b^1} \leq \delta$. The support of $EF$ lies in $[a^1]\times [b^1]$ and we divide this into four regions, as shown in Figure \ref{fig:WVsupport}. Region $2$ has  
the property that $e<w$ always, and region $3$ has the property that $f<v$ always. Thus we can apply Lemma \ref{lem:1} to random variables $FV$ for region $2$ and $EW$ for region $3$. No argument is needed in region $4$, as all $(e,f)$ satisfy the property that $e<w$ and $v<f$. The issue lies with region $1$, where we do not know how to show the desired result. Fortunately, the probability that $EF$ lies in region $1$ is at most $\delta$, by the choice of $a^0, a^1, b^0$ and $b^1$ and using the fact that    $E \mid (W=w)$ is uniform in $[w]$ and $F\mid (V=v)$ is uniform in $[v]$.

Hence, we divide the square $[K]\times [K]$ into several smaller squares such that for any square $\{a_1, \ldots b_1\}\times \{a_2, \ldots b_2\}$, it holds that $\frac{(b_1-a_1)(b_2-a_2)}{b_1b_2} \leq \delta$. This leads to the recursive decomposition given in Figure~\ref{fig:rectdecompose}. To keep the number of squares less than $\approx\log\max\{|\M|, |\N|\}$, we perturb the distribution of $MNXYZ$ such that $p_{M| X=x}(m) \geq \frac{\delta}{|\M|}$ and $p_{N| Y=y}(n) \geq \frac{\delta}{|\N|}$. This is possible with an error of at most $\delta$. Finally, we construct a set that contains a large support of $EF$ for every square in the decomposition given in Figure~\ref{fig:rectdecompose}. By taking a union over all these sets, we obtain the desired result, with a loss of $\log\log\max\{|\M|, |\N|\}$ that is reflected in the statement of Theorem~\ref{thm:main}. 

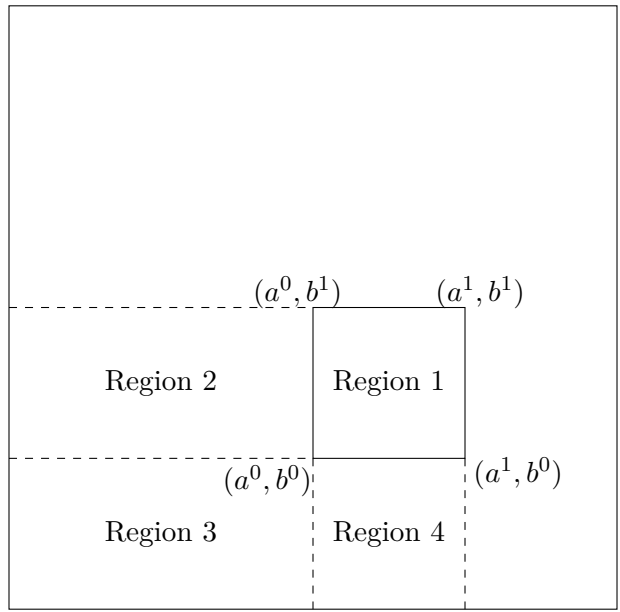
\begin{figure}[ht]
\centering
\begin{tikzpicture}[xscale=1,yscale=1]

\draw (0,0) rectangle (8,8);
\draw (4,2) rectangle (6,4);
\node at (3.4, 1.7) {$(a^0, b^0)$};
\node at (6.2, 4.2) {$(a^1, b^1)$};
\node at (3.8, 4.2) {$(a^0, b^1)$};
\node at (6.7, 1.8) {$(a^1, b^0)$};
\draw [dashed] (0,2) -- (4,2);
\draw [dashed] (0,4) -- (4,4);
\draw [dashed] (4,2) -- (4,0);
\draw [dashed] (6,2) -- (6,0);

\node at (5,3) {Region 1};
\node at (2,3) {Region 2};
\node at (2,1) {Region 3};
\node at (5,1) {Region 4};

\end{tikzpicture}
\caption{\small Suppose the joint distribution of $WV$ lies in Region $1$  shown above, with $\frac{(a^1-a^0)(b^1-b^0)}{a^1b^1} \leq \delta$. Then the probability that $EF$ lies within Region $1$ is at most $\delta$, as given any $w,v$, the random variable $EF\mid (WV=w,v)$ is uniformly distributed in $[w]\times [v]$. Any $(e,f)$ in Region $3$ automatically satisfies that $e<w$ and $f<v$. For Region $2$, we use an argument similar to Lemma~\ref{lem:1} for the random variable $FV$ and for Region $4$ we do the same for the random variable $EW$.}
 \label{fig:WVsupport}
\end{figure}

\begin{proof}[Proof of Lemma~\ref{lem:setA}]
We set $dev\defeq\log\frac{\max\set{\abs{\M},\abs{\N}}}{\delta}$ for convenience. For any $\br{m,n}\in\M\times\N$, we define $w_m\br{x}\defeq Kp_{M|X=x}\br{m}$ and $v_n\br{y}\defeq Kp_{N|Y=y}\br{n}$. From our assumption on $K$, $w_m\br{\cdot}$ and $v_n\br{\cdot}$ are both integer-valued functions for all $m, n$.  Further define
\[\mathrm{Good}_{m,n,z}^1\defeq\set{\br{x,y}:\frac{p_{M|X=x}\br{m}}{p_{M|N=n, Z=z}\br{m}}\leq\delta\cdot 2^{R_1},\frac{p_{N|Y=y}\br{n}}{p_{N|M=m, Z=z}\br{n}}\leq\delta\cdot 2^{R_2}\atop~\quad\mbox{and}~\frac{p_{M|X=x}\br{m}p_{N|Y=y}\br{n}}{p_{MN\mid Z=z}\br{m,n}}\leq\frac{\delta}{dev}2^{R_1+R_2}},\]
\[\mathrm{Good}^2_{m,n, z}\defeq\set{\br{x,y}:p_{M|X=x}\br{m}\geq\frac{\delta}{\abs{\M}}~\mbox{and}~p_{N|Y=y}\br{n}\geq\frac{\delta}{\abs{\N}}}\]
and
\[\mathrm{Good}_{m,n,z}\defeq\mathrm{Good}^1_{m,n,z}\cap\mathrm{Good}^2_{m,n,z}.\]
We define the new random variables $X'Y'E'F'MNZ$ obtained by restricting $XY$ to $\mathrm{Good}_{MNZ}$. Namely,
\begin{eqnarray*}     
&&p_{X'Y'E'F'\mid MNZ=mnz}\br{x,y,e,f}\defeq\\
&&\begin{cases}
\frac{p_{XY|MNZ=mnz}\br{x,y}}{p_{XY|MNZ=mnz}\br{\mathrm{Good}_{m,n,z}}w_m\br{x}v_n\br{y}}&\mbox{if $\br{x,y}\in\mathrm{Good}_{m,n,z} \wedge e\leq w_m\br{x}\wedge f\leq v_n\br{y}$}
\\
0~&\mbox{otherwise.}\end{cases}
\end{eqnarray*}
Here, $\wedge$ refers to `and'. Let $WV$ be jointly correlated with $X'Y'E'F'MNZ$ and defined as 
\begin{eqnarray*}     
p_{WV\mid X'Y'E'F'MNZ=xyefmnz}\br{w,v}\defeq\begin{cases}
1 &\mbox{if $w= w_m\br{x}\wedge v= v_n\br{y}$}
\\
0~&\mbox{otherwise.}\end{cases}
\end{eqnarray*}
From Fact~\ref{fac:distancerestriction} it holds for any $m,n,z$ that
\begin{eqnarray}
&&\onenorm{\br{E'F'X'Y'|MNZ=mnz}-\br{EFXY|MNZ=mnz}}\nonumber\\ &&=\onenorm{\br{X'Y'|MNZ=mnz}-\br{XY|MNZ=mnz}}\nonumber\\
&&=2\br{1-\prob{XY|MNZ=mnz}{\mathrm{Good}_{m,n,z}}}.\label{eqn:good}
\end{eqnarray}

\begin{figure}[ht]
\centering
\begin{tikzpicture}[xscale=1,yscale=1]

\draw (0,0) rectangle (10,10);
\draw (0,0) rectangle (4,4);
\draw (0,0) rectangle (1.6,1.6);

\draw [dashed] (10,8) -- (0,8);
\draw [dashed] (10,6) -- (0,6);
\draw [dashed] (10,4) -- (0,4);
\draw [dashed] (10,2) -- (4,2);

\draw [dashed] (8,10) -- (8,0);
\draw [dashed] (6,10) -- (6,0);
\draw [dashed] (4,10) -- (4,0);
\draw [dashed] (2,10) -- (2,4);

\draw [dashed] (4,3.2) -- (0,3.2);
\draw [dashed] (4,2.4) -- (0,2.4);
\draw [dashed] (4,1.6) -- (0,1.6);
\draw [dashed] (4,0.8) -- (1.6,0.8);

\draw [dashed] (3.2,4) -- (3.2,0);
\draw [dashed] (2.4,4) -- (2.4,0);
\draw [dashed] (1.6,4) -- (1.6,0);
\draw [dashed] (0.8,4) -- (0.8,1.6);

\node at (-0.2,10) {$K$};
\node at (10, -0.2) {$K$};
\node at (-0.3,4) {$\delta_1 K$};
\node at (4, -0.3) {$\delta_1 K$};
\node at (-0.3,1.6) {$\delta_1^2 K$};
\node at (1.6, -0.3) {$\delta_1^2 K$};

\node at (-0.7,8) {$K-\delta K$};
\node at (8, -0.3) {$K-\delta K$};
\node at (-1,3.2) {$\delta_1K - \delta^2 K$};

\node at (0.75, 0.75) {$S_c$};

\end{tikzpicture}
\caption{\small The partition in the proof of Lemma~\ref{lem:setA} : The square $[K]\times [K]$ is divided into a collection of squares as depicted above. The squares with bold boundaries are constructed by scaling $[K]\times [K]$ by integral powers of $\delta_1$. This leads to a self similar collection of $6$-sided polygons with bold boundaries. The squares with dashed boundaries further decompose each such $6$-sided polygon. The innermost square $S_c$ is not decomposed to keep the number of decompositions small. $S_c$ does not contain the support of $WV$ due to our restriction of $XY$ to the set $\mathrm{Good}^2$.}
 \label{fig:rectdecompose}
\end{figure}
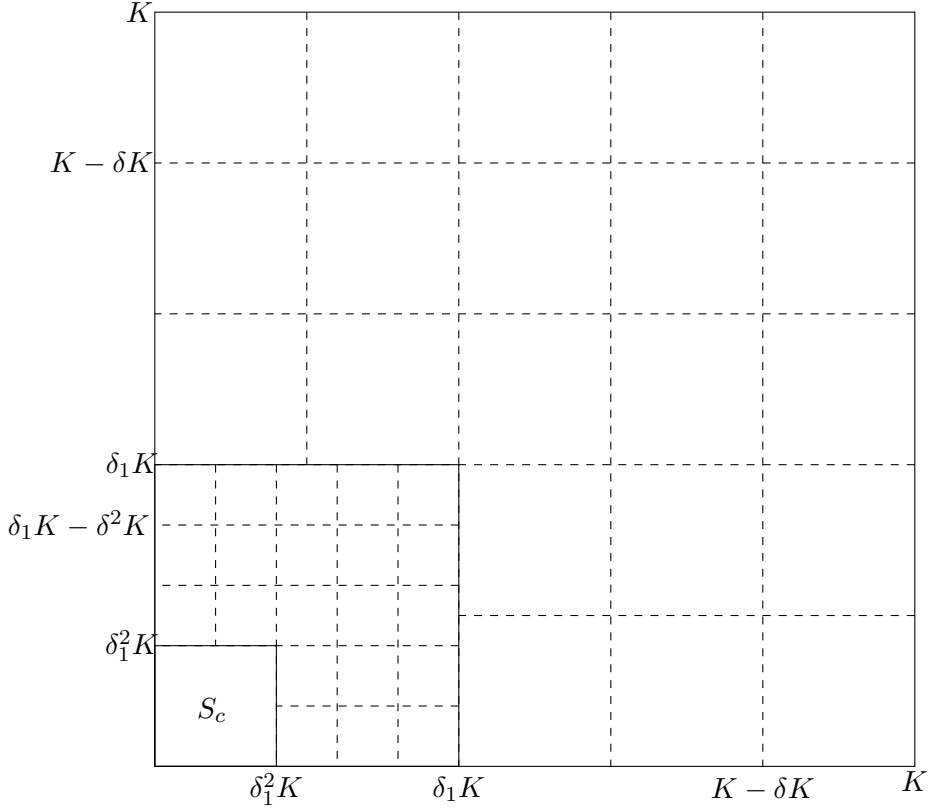

\noindent\textbf{Constructing a partition of $[K]\times[K]$:} The idea behind the construction is depicted in Figure~\ref{fig:rectdecompose}. Let $\delta_1\defeq\sqrt{\delta}$ and $c\defeq\frac{1}{\log\frac{1}{\delta_1}}\cdot\log\frac{\max\set{\abs{\M},\abs{\N}}}{\delta}$. We may assume that $\delta_1^c K$ is an integer.  Let $\set{C_j}_{j=1}^{\frac{1}{\delta^2}}$ be a set of squares of size $\delta K\times\delta K$ that partition the set $[K]\times [K]$. Let $\alpha C_j$ denote the set $\set{\br{w,v}:\br{\frac{w}{\alpha},\frac{v}{\alpha}}\in C_j}$. For all $i\in\set{0,\ldots,c}$, we define

\[S_i\defeq [\delta_1^iK]\times [\delta_1^iK], T_i\defeq S_i\setminus S_{i+1}.\]
with $T_c$ being undefined. For $i\in\set{0,1,\ldots, c}$ and $j\in[\frac{1}{\delta^2}]$

\[T_{i,j}\defeq T_i\cap\delta_1^iC_j.\]
Then $\set{T_{i,j}}_{i\in\set{0,1,\ldots, c},j\in[\frac{1}{\delta^2}]}$ are disjointed sets covering $[K]\times[K]\setminus S_c$ and $\set{\delta_1^i C_j}_{j=1}^\frac{1}{\delta^2}$ equally divide $S_i$ into $\frac{1}{\delta^2}$ squares.  Thus, $T_{i,j}$ is either a $\delta_1^i\delta K\times\delta_1^i\delta K$ square or empty. For any $i$, $T_{i,j}$ is non-empty for $\br{\frac{1-\delta_1}{\delta}}^2$ many $j$'s. We relabel the indices $j$'s such that $T_{i,j}$ is non-empty if and only if $1\leq j\leq \br{\frac{1-\delta_1}{\delta}}^2$. Set $T_{i,j}\defeq\set{a_{i,j}^0,a_{i,j}^0+1,\ldots, a_{i,j}^1}\times\set{b_{i,j}^0,a_{i,j}^0+1,\ldots, b_{i,j}^1}$. We have $a_{i,j}^1-a_{i,j}^0=b_{i,j}^1-b_{i,j}^0=\delta_1^i\delta K$ and $a_{i,j}^0\geq\delta_1^{i+1}K$ and $b_{i,j}^0\geq\delta_1^{i+1}K$ as $T_{i,j}$ and $S_{i+1}$ are disjoint. Let $\T_{i,j}$ denote the event that $\br{W,V}\in T_{i,j}$. Note that for any $\br{w,v}\in\supp{W}\times\supp{V}$, $w\geq K\cdot\frac{\delta}{\abs{\M}}$ and $v\geq K\cdot\frac{\delta}{\abs{\N}}$. From the choice of $c$, we have $\delta_1^cK=K\cdot\frac{\delta}{\max\set{\abs{\M}, \abs{\N}}}$, which implies that $S_c$ does not intersect the support of $\br{WV|MNZ=mnz}$ for any $m,n,z$. Given $i,j,m,n,z$, we set 

\begin{eqnarray*}
&&\mathrm{Bad}_{i,j,m,n,z}\defeq\\
&&\set{\br{e,f}\in\supp{E'F'|MNZ=mnz,\T_{i,j}}\setminus T_{i,j}:\sum_{\br{w,v}\in T_{i,j:w\geq e\wedge v\geq f}}p_{WV|MNZ=mnz,\T_{i,j}}\br{w,v}\leq\delta.}.
\end{eqnarray*}

\noindent\textbf{Constructing $\A$:}. For any $i, j, m, n, z$, set $$\A_{i,j,m,n,z}\defeq\supp{E'F'|MNZ=mnz,\T_{i,j}}\setminus\br{\T_{i,j}\cup\mathrm{Bad}_{i,j,m,n,z}}.$$ We conclude from Claim~\ref{claim:tij} and Claim~\ref{claim:bad} that 
\[\sum_{\br{e,f}\in\A_{i,j,m,n,z}}p_{E'F'|MNZ=mnz,\T_{i,j}}\br{e,f}\geq1-3\delta.\]
Moreover, for all $\br{e,f}\in\A_{i,j,m,n,z}$, which ensures that $\br{e,f}\notin\mathrm{Bad}_{i,j,m,n,z}$, we have 

\begin{eqnarray}
\label{eq:setAupperbound}
&&p_{E'F'|MNZ=mnz,\T_{i,j}}\br{e,f}\nonumber\\
&=&\sum_{\br{w,v}\in T_{i,j}}p_{WV|MNZ=mnz,\T_{i,j}}\br{w,v}\cdot p_{E'F'|WVMNZ=wvmnz,\T_{i,j}}\br{e,f}\nonumber\\
&=&\sum_{\br{w,v}\in\T_{i,j}:w\geq e\wedge v\geq f}p_{WV|MNZ=mnz,\T_{i,j}}\br{w,v}\cdot\frac{1}{wv}\nonumber\\
&\geq&\frac{dev}{\delta^4K^22^{R_1+R_2}p_{MN\mid Z=z}\br{m,n}}\cdot\sum_{\br{w,v}\in\T_{i,j}:w\geq e\wedge v\geq f}p_{WV|MNZ=mnz,\T_{i,j}}\br{w,v}\nonumber\\
&\geq&\frac{dev}{\delta^3K^22^{R_1+R_2}p_{MN\mid Z=z}\br{m,n}},
\end{eqnarray}
where we use the fact that for all $\br{x,y}\in\mathrm{Good}_{m,n}$
\[w\br{x}\cdot v\br{y}=K^2p_{M|X=x}\br{m}p_{N|Y=y}\br{n}\leq K^2\cdot\frac{\delta^4}{dev}\cdot 2^{R_1+R_2}p_{MN\mid Z=z}\br{m,n}.\] 
Summing Eq.~\eqref{eq:setAupperbound} over all $\br{e,f}\in\A_{i,j,m,n,z}$, we have 
\[1\geq\frac{dev\abs{A_{i,j,m,n,z}}}{\delta^3K^22^{R_1+R_2}p_{MN\mid Z=z}\br{m,n}},\]
which implies that 
\[\abs{\A_{i,j,m,n,z}}\leq\frac{\delta^3K^22^{R_1+R_2}p_{MN\mid Z=z}\br{m,n}}{dev}.\]
We further define
\begin{eqnarray*}
&&\A^{(1)}_{m,n,z}\defeq\cup_{i,j}\A_{i,j,m,n,z},\\
&&\A^{(2)}_{m,n,z}\defeq\set{\br{e,f}:e\leq\delta K2^{R_1}p_{M|N=n,Z=z}\br{m}},\\
&&\A^{(3)}_{m,n,z}\defeq\set{\br{e,f}:f\leq\delta K2^{R_1}p_{M|N=n,Z=z}\br{n}},\\
&&\A\defeq\set{\br{m,n,z,e,f}:\br{e,f}\in\A^{(1)}_{m,n,z}\cap\A^{(2)}_{m,n,z}\cap\A^{(3)}_{m,n,z}}.
\end{eqnarray*}
Then
\begin{eqnarray}
\abs{\A^{\br{1}}_{m,n,z}}&\leq&\sum_{i,j}\abs{\A_{i,j,m,n,z}}\leq\br{c+1}\cdot\frac{1}{\delta^2}\cdot\frac{\delta^3K^22^{R_1+R_2}p_{MN\mid Z=z}\br{m,n}}{dev}\nonumber\\
&\leq&\delta K^22^{R_1+R_2}p_{MN\mid Z=z}\br{m,n}.\label{eqn:a1upperbound}
\end{eqnarray}
And
\begin{eqnarray}
\prob{E'F'|MNZ=mnz}{\A_{m,n,z}^{\br{1}}}&=&\sum_{i,j}\prob{WV|MNZ=mnz}{T_{i,j}}\prob{WV|MNZ=mnz,\T_{i,j}}{\A_{m,n,z}^{\br{1}}}\nonumber\\
&\geq&\sum_{i,j}\prob{WV|MNZ=mnz}{T_{i,j}}\prob{WV|MNZ=mnz,\T_{i,j}}{\A_{i,j,m,n,z}}\nonumber\\
&\geq&1-3\delta.\label{eqn:a1}
\end{eqnarray}
Note that for any $m,n,z$
\[\prob{E'F'|MNZ=mnz}{\A^{\br{2}}_{m,n,z}}=\prob{E'F'|MNZ=mnz}{\A^{\br{3}}_{m,n,z}}=1,\]
which implies that 
\[\prob{E'F'|MNZ=mnz}{\A^{\br{1}}_{m,n,z}\cap\A^{\br{2}}_{m,n,z}\cap\A^{\br{3}}_{m,n,z}}\geq1-3\delta.\]
Combining with Eq.~\eqref{eqn:good}, we have
\[\prob{EF|MNZ=mnz}{\A^{\br{1}}_{m,n,z}\cap\A^{\br{2}}_{m,n,z}\cap\A^{\br{3}}_{m,n,z}}\geq \prob{XY|MNZ=mnz}{\mathrm{Good}_{m,n,z}}-3\delta.\]
Therefore,
\[\prob{MNZEF}{\A}\geq\sum_{m,n,z}p_{MNZ}\br{m,n,z}\prob{XY|MNZ=mnz}{\mathrm{Good}_{m,n,z}}-3\delta\geq1-\epsilon-5\delta,\]
where the last inequality follows from Claim~\ref{claim:good}.  Thus, we conclude the first inequality in Lemma~\ref{lem:setA}. For the second inequality in Lemma~\ref{lem:setA}, consider

\begin{eqnarray*}
\prob{MEZ\times T}{\A}&=&\sum_{m,n}p_{MZ}\br{m,z}p_{\tilde{T}}\br{n}\prob{\br{E|MZ=mz}\times L}{A^{\br{1}}_{m,n,z}\cap A^{\br{2}}_{m,n,z}\cap A^{\br{3}}_{m,n,z}}\\
&\leq&\sum_{m,n,z}p_{MZ}\br{m,z}p_{\tilde{T}}\br{n}\prob{\br{E|MZ=mz}\times L}{A^{\br{3}}_{m,n,z}}\\
&=&\sum_{m,n,z}p_{MZ}\br{m,z}p_{\tilde{T}}\br{n}\sum_{\br{e,f}:f\leq\delta K2^{R_2}p_{N|MZ=mz}\br{n}}\frac{1}{K}p_{E|MZ=mz}\br{e}\\
&=&\sum_{m,n,z}p_{MZ}\br{m,z}\frac{1}{\abs{\N}}\sum_e\delta 2^{R_2}p_{N|MZ=mz}\br{n}p_{E|MZ=mz}\br{e}\\
&=&\frac{\delta 2^{R_2}}{\abs{\N}}\sum_{m,n,z}p_{MZ}\br{m,z}p_{N|MZ=mz}\br{n}\leq\frac{\delta 2^{R_2}}{\abs{\N}}.
\end{eqnarray*}
Similarly,
\[\prob{S\times NFZ}{\A}\leq\frac{\delta2^{R_1}}{\abs{\M}}.\]
For the last inequality, we apply Eq.~\eqref{eqn:a1upperbound} to conclude,
\begin{eqnarray*}
\prob{S\times T\times Z}{\A}&=&\frac{1}{\abs{\M}\abs{\N}}\sum_{m,n, z}p_Z(z)\prob{L\times L}{\A^{(1)}_{m,n,z}}\leq\sum_{m,n,z}p_Z(z)\frac{\abs{\A^{\br{1}}_{m,n,z}}}{K^2\abs{\M}\abs{\N}}\leq\frac{\delta 2^{R_1+R_2}}{\abs{\M}\abs{\N}}.
\end{eqnarray*}
\end{proof}

The following claims were used in the above lemma.

\begin{claim}\label{claim:good}
$\sum_{m,n,z}p_{MNZ}\br{m,n,z}\prob{XY|MNZ=mnz}{\mathrm{Good}_{m,n,z}}\geq1-\epsilon-2\delta$.
\end{claim}
\begin{proof}
From the choice of $\epsilon$ in Theorem~\ref{thm:main}, we have

\begin{eqnarray*}
&&\sum_{m,n,z}p_{MNZ}\br{m,n,z}\prob{XY|MNZ=mnz}{\mathrm{Good}_{m,n,z}}\\
&&\geq 1-\sum_{m,n,z,x,y}p_{MNZXY}\br{m,n,z,x,y}\br{\id\br{\br{x,y}\notin\mathrm{Good}^1_{m,n,z}}+\id\br{\br{x,y}\notin\mathrm{Good}_{m,n,z}^2}}\\
&&\geq 1-\epsilon-\sum_{m,n,z,x,y}p_{MNZXY}\br{m,n,z,x,y}\cdot\id\br{\br{x,y}\notin\mathrm{Good}_{m,n,z}^2}\\
&&= 1-\epsilon-\sum_{x,y,z}p_{XYZ}\br{x,y,z}\sum_{m,n}p_{M|X=x}\br{m}p_{N|Y=y}\br{n}\\
&&\hspace{2cm}\id\br{p_{M|X=x}\br{m}\leq\frac{\delta}{\abs{\M}}~\text{or}~p_{N|Y=y}\br{n}\leq\frac{\delta}{\abs{\N}}}\\
&&\geq 1-\epsilon-\sum_{x,y,z}p_{XYZ}\br{x,y,z}\sum_{m,n}p_{M|X=x}\br{m}p_{N|Y=y}\br{n}\id\br{p_{M|X=x}\br{m}\leq\frac{\delta}{\abs{\M}}}\\
&& -\sum_{x,y,z}p_{XYZ}\br{x,y,z}\sum_{m,n}p_{M|X=x}\br{m}p_{N|Y=y}\br{n}\id\br{p_{N|Y=y}\br{n}\leq\frac{\delta}{\abs{\N}}}\\
&&\geq 1-\epsilon-2\delta.
\end{eqnarray*}

\end{proof}

The following claim upper bounds the Region 1 in Figure~\ref{fig:WVsupport}.
\begin{claim}\label{claim:tij}
For any $i,j$, it holds that 
\[\sum_{\br{e,f}\in T_{i,j}}p_{E'F'|MNZ=mnz,\T_{i,j}}\br{e,f}\leq\delta.\]
\end{claim}
\begin{proof}
Note that
\begin{eqnarray*}
&&\sum_{\br{e,f}\in T_{i,j}}p_{E'F'|MNZ=mnz,\T_{i,j}}\br{e,f}\\
&&=\sum_{\br{e,f}\in T_{i,j}}\sum_{\br{W,V}\in T_{i,j}:w\geq e\wedge v\geq f}p_{WV|MNZ=mnz, \T_{i,j}}\br{w,v}\cdot\frac{1}{wv}\\
&&=\sum_{\br{w,v}\in T_{i,j}}p_{WV|MNZ=mnz,\T_{i,j}}\br{w,v}\cdot\sum_{\br{e,f}\in T_{i,j}:w\geq e\wedge v\geq f}\frac{1}{wv}\\
&&=\sum_{\br{w,v}\in T_{i,j}}p_{WV|MNZ=mnz,\T_{i,j}}\br{w,v}\cdot\frac{\br{w-a_{i,j}^0}\br{v-b_{i,j}^0}}{wv}\\
&&\leq \sum_{\br{w,v}\in T_{i,j}}p_{WV|MNZ=mnz,\T_{i,j}}\br{w,v}\cdot\frac{\br{a_{i,j}^1-a_{i,j}^0}\br{b_{i,j}^1-b_{i,j}^0}}{a_{i,j}^1b_{i,j}^1}\\
&&\leq \frac{\delta_1^{2i}\delta^2K^2}{\delta_1^{2i+2}K^2}=\delta,
\end{eqnarray*}
where the first inequality follows from the definition of $T_{i,j}$.
\end{proof}

The following claim upper bounds the error in Regions $2,3$ and $4$ in Figure~\ref{fig:WVsupport}.
\begin{claim}\label{claim:bad}
For any $i,j,m,n,z$, it holds that 
\[\sum_{\br{e,f}\in\mathrm{Bad}_{i,j,m,n,z}}p_{E'F'|MNZ=mnz,\T_{i,j}}\br{e,f}\leq 2\delta.\]
\end{claim}

\begin{proof}

Let $w^*, v^*$ be the smallest integers such that 
\[\sum_{\br{w,v}\in T_{i,j}:w\geq w^*}p_{WV|MNZ=mnz,\T_{i,j}}\br{w,v}\leq\delta,\]
and
\[\sum_{\br{w,v}\in T_{i,j}:v\geq v^*}p_{WV|MNZ=mnz,\T_{i,j}}\br{w,v}\leq\delta.\]
We claim that  for any $\br{e,f}\in\mathrm{Bad}_{i,j,m,n,z}$, either $e\geq w^*$ or $f\geq v^*$. Suppose by contradiction that there exists $\br{e,f}\in\mathrm{Bad}_{i,j,m,n,z}$ such that $e\leq w^*$ and $f\leq v^*$. As $\br{e,f}\notin T_{i,j}$, either $e\leq a_{i,j}^0$ or $f\leq b_{i,j}^0$. Suppose $e\leq a_{i,j}^0$ (corresponding to Region $2$ in Figure \ref{fig:WVsupport}; the other case follows similarly). Then
\[\sum_{\br{w,v}\in T_{i,j}: w\geq e\wedge v\geq f}p_{WV|MNZ=mnz,\T_{i,j}}\br{w,v}=\sum_{\br{w,v}\in T_{i,j}:v\geq f}p_{WV|MNZ=mnz,\T_{i,j}}\br{w,v}>\delta,\]
where the equality is from the fact that $\supp{WV|MNZ=mnz,\T_{i,j}}\subseteq\set{a_{i,j}^0,\ldots, a_{i,j}^1}\times\set{b_{i,j}^0,\ldots, b_{i,j}^1}$ which ensures that every $w$ is larger than $e$;  the inequality follows from the definition of $f^*$.
Therefore, we have 
\[\sum_{\br{e,f}\in\mathrm{Bad}_{i,j,m,n,z}}p_{E'F'|MNZ=mnz,\T_{i,j}}\br{e,f}\leq\br{\sum_{\br{e,f}:e\geq w^*}+\sum_{\br{e,f}:f\geq v^*}}p_{E'F'|MNZ=mnz,\T_{i,j}}\br{e,f}.\]
We upper bound the first summation on the right hand side, following Lemma \ref{lem:1}.
\begin{eqnarray*}
&&\sum_{\br{e,f}:e\geq w^*}p_{E'F'|MNZ=mnz,\T_{i,j}}\br{e,f}\\
&&= p_{E'F'|MNZ=mnz,\T_{i,j}}\br{e,f}\sum_{\br{w,v}\in T_{i,j}:w\geq e\wedge v\geq f}p_{WV|MNZ=mnz,\T_{i,j}}\br{w,v}\cdot\frac{1}{wv}\\
&&= \sum_{\br{w,v}\in T_{i,j}:w\geq w^*\wedge v}p_{WV|MNZ=mnz,\T_{i,j}}\br{w,v}\cdot\sum_{\br{e,f}:w\geq e\geq w^*\wedge f\leq v}\frac{1}{wv}\\
&&\leq \sum_{\br{w,v}\in T_{i,j}:w\geq w^*\wedge v}p_{WV|MNZ=mnz,\T_{i,j}}\br{w,v}\leq\delta.
\end{eqnarray*}
Similarly,
\[\sum_{\br{e,f}:f\geq v^*}p_{E'F'|MNZ=mnz,\T_{i,j}}\br{e,f}\leq\delta.\]
Thus, we conclude the claim.
\end{proof}

\subsection{Compression in terms of conditional mutual information}

In this subsection, we present a simpler feasible communication region in terms of conditional mutual information, which is obtained via arguments similar to the Substate theorem in~\cite{Jain:2003}.

\begin{theorem}\label{maintask1}
Given a joint distribution $XYZMN$ over $\X\times\Y\times\Z\times\M\times\N$ satisfying  $M-X-YNZ, XMZ-Y-N$ and any two reals $R_1,R_2\geq 0$ satisfying 
\begin{enumerate}
\item $R_1\geq\condmutinf{X}{M}{NZ}$;
\item $R_2\geq\condmutinf{Y}{N}{MZ}$;
\item $R_1+R_2\geq\condmutinf{XY}{MN}{Z}$;
\end{enumerate}
suppose Alice and Bob are given $x$ and $y$, respectively, where $\br{x,y}$ are drawn from distribution $XY$.  Then for any $\delta\in (0,1)$ there exists a $$\br{\frac{16 R_1}{\delta^2}+\frac{10}{\delta}+\log\log\frac{\max\set{\abs{\M},\abs{\N}}}{\delta},  \frac{16 R_2}{\delta^2}+\frac{10}{\delta}+\log\log\frac{\max\set{\abs{\M},\abs{\N}}}{\delta}, \delta}$$ two-senders-one-receiver message compression with side information at the receiver. 

\end{theorem}
\begin{proof}
Let $\delta'>0$ be a parameter chosen later. Note that 
$$\condmutinf{X}{M}{NZ}=\expec{n,z,x\sim NZX}{\relent{M|X=x}{M|NZ=nz}}.$$ 
By the Markov inequality,
\[\prob{n,z,x\sim NZX}{\relent{\br{M|X=x}}{\br{M|NZ=nz}}\leq \frac{R_1}{\delta'}}\geq1-\delta'.\]
For any $\br{n,z,x}$ satisfying $\relent{\br{M|X=x}}{\br{M|NZ=nz}}\leq \frac{R_1}{\delta'}$, we have 
\begin{eqnarray*}
\frac{R_1}{\delta'}&\geq&\sum_mp_{M|X=x}\br{m}\log\frac{p_{M|X=x}\br{m}}{p_{M|NZ=nz}\br{m}}\\
&\geq&\prob{m\sim\br{M|X=x}}{\frac{p_{M|X=x}\br{m}}{p_{M|NZ=nz}\br{m}}\leq2^{\frac{R_1/\delta'+1}{\delta'}}}\cdot\frac{R_1/\delta'+1}{\delta'}-1,
\end{eqnarray*}
where the inequality follows from the fact that $\sum_ia_i\log\frac{a_i}{b_i}\geq -1$ if $a_i,b_i\geq 0, \sum_ia_i\leq 1$ and $\sum_ib_i\leq 1$.  It implies that 
\[\prob{m\sim\br{M|X=x}}{\frac{p_{M|X=x}\br{m}}{p_{M|NZ=nz}\br{m}}\leq2^{\frac{R_1/\delta'+1}{\delta'}}}\cdot\frac{R_1/\delta'+1}{\delta'}\leq\delta'.\]
Therefore,
\[\prob{x,m,n,z\sim XMNZ}{\frac{p_{M|X=x}\br{m}}{p_{M|NZ=nz}\br{m}}}=\prob{x,m,n,z\sim XMNZ}{\frac{p_{M|X=x}\br{m}}{p_{M|NZ=nz}\br{m}}\leq2^{\frac{R_1/\delta'+1}{\delta'}}}\geq1-2\delta',\]
where the equality follows from the fact that $M-X-NZ$. Similarly,
\[\prob{y,m,n,z\sim YMNZ}{\frac{p_{N|Y=y}\br{n}}{p_{N|MZ=mz}\br{n}}\leq2^{\frac{R_2/\delta'+1}{\delta'}}}\geq1-2\delta',\]
and
\[\prob{x,y,m,n,z\sim XYMNZ}{\frac{p_{MN|XY=xy}\br{m,n}}{p_{MN\mid Z=z}\br{m,n,z}}\leq2^{\frac{\br{R_1+R_2}/\delta'+1}{\delta'}}}\geq1-2\delta'.\]
Applying Theorem~\ref{thm:main}, by setting $\delta'$ to be the largest real number less than $\delta/10$ satisfying that $\frac{1}{\sqrt{\delta'}}$ is an integer, $R_1\rightarrow \frac{16 R_1}{\delta^2}+\frac{10}{\delta}+\log\log\frac{\max\set{\abs{\M},\abs{\N}}}{\delta}$, $R_2\rightarrow \frac{16 R_2}{\delta^2}+\frac{10}{\delta}+\log\log\frac{\max\set{\abs{\M},\abs{\N}}}{\delta}$, we conclude the result.

\end{proof}

\subsection{Comparision with the bound obtained in~\cite{AnshuJW17un}}
\label{subsec:comp}

In~\cite[Theorem 4]{AnshuJW17un}, the following achievable communication region was obtained for the task in Definition~\ref{def:task1} (the authors also state a more general bound optimized over all possible extensions of $MN$, but that involves auxiliary random variables of unbounded size):
\begin{align}
\label{eq:unifiedach}
R_1 &\geq \misdiv{\delta}{XM}{X\times S} -  \htisdiv{\epsilon_1}{MNZ}{S\times NZ} + 4\log \frac{3}{\delta} , \nonumber\\
R_2 &\geq \misdiv{\delta}{YN}{Y\times T} - \htisdiv{\epsilon_2}{MZN}{MZ\times T}  + 4\log \frac{3}{\delta} ,\nonumber\\
R_1+R_2 &\geq \misdiv{\delta}{XM}{X\times S} + \misdiv{\delta}{YN}{Y\times T} - \htisdiv{\epsilon_3}{MNZ}{S\times T\times Z} + 6\log \frac{3}{\delta},
\end{align}
giving the overall error of $\epsilon_1+\epsilon_2+\epsilon_3+13\delta$. Above, $S$ and $T$ are arbitrary random variables over $\M$ and $\N$, respectively. The following claim shows that this achievable communication region is contained inside the achievable communication region of Theorem~\ref{thm:main}, up to the additive factor of $\log\log\max\br{|\M|, |\N|}$.

\begin{theorem}
For any $(R_1, R_2)$ satisfying~Eqs~\eqref{eq:unifiedach}. It holds that 
$$\prob{x,y,z,m,n\sim XYZMN}{\frac{p_{M|X=x}(m)}{p_{M|NZ=n,z}(m)}\leq \frac{\delta^4}{3^4}2^{R_1} \mbox{ and } \frac{p_{N|Y=y}(n)}{p_{N|MZ=m,z}(n)}\leq  \frac{\delta^4}{3^4}2^{R_2}\atop \mbox{ and } \frac{p_{M|X=x}(m)p_{N|Y=y}(n)}{p_{MN|Z=z}(m,n)}\leq \frac{\delta^6}{3^6}2^{R_1+R_2}} \geq 1-\epsilon_1-\epsilon_2-\epsilon_3-2\delta.$$

\end{theorem}
\begin{proof}
Let $k_1\defeq \misdiv{\delta}{XM}{X\times S}$, $k_2\defeq\htisdiv{\epsilon_1}{MNZ}{S\times NZ}$, 
$k_3\defeq \misdiv{\delta}{YN}{Y\times T}$, $k_4\defeq \htisdiv{\epsilon_2}{MZN}{MZ\times T}$ and $k_5\defeq \htisdiv{\epsilon_3}{MNZ}{S\times T\times Z}$. By the union bound, we can find a subset $\S\subseteq \X\times \Y\times\Z\times \M\times\N$ such that $p_{XYZMN}\br{\S}\geq 1-\epsilon_1-\epsilon_2-\epsilon_3-2\delta$ and for all $(x,y,z,m,n)\in \S$
$$\frac{p_{M|X=x}(m)}{p_S(m)}\leq 2^{k_1}, \frac{p_{M|NZ=n,z}(m)}{p_S(m)}\geq 2^{k_2}, \frac{p_{N|Y=y}(n)}{p_T(n)}\leq 2^{k_3},$$  $$\frac{p_{N|MZ=m,z}(n)}{p_T(n)}\geq 2^{k_4}, \frac{p_{MN|Z=z}(m,n)}{p_S(m)p_T(n)}\geq 2^{k_5},$$
which together imply that
$$\frac{p_{M|X=x}(m)}{p_{M|NZ=n,z}(m)}\leq 2^{k_1-k_2}, \frac{p_{N|Y=y}(n)}{p_{N|MZ=m,z}(n)}\leq 2^{k_3-k_4}, \frac{p_{M|X=x}(m)p_{N|Y=y}(n)}{p_{MN|Z=z}(m,n)}\leq 2^{k_1+k_3-k_5}.$$
Substituting the values of $R_1, R_2$ the proof concludes.
\end{proof}

\section{Consequences of Theorem~\ref{thm:main}}
\label{sec:cons}

In this section, we present several consequences of our main theorem.

\subsection{Lossy distributed source coding}

Lossy source coding is a well studied task in information theory~\cite{WynerZ76, IwataM02, WatanabeKT15, YasAG13, KostinaV12, Kostina13}, where a sender observes a sample from a source and the receiver is allowed to output a distorted version of this sample. Our first application is for the problem of lossy distributed source coding~\cite{Tung78, Berger78, Berger89, Wagner11} (which is a distributed version of the lossy source coding task), which is defined as follows (observe that we also include a size information with the receiver in our definition below).

\begin{definition}
	\textbf{A $\br{k,\epsilon}$-lossy distributed source coding with side information.} Given $k>0$ and $\epsilon\in(0,1)$ and joint random variables $XYZ$ over $\X\times\Y\times\Z$, Alice, Bob and Charlie observe $X, Y$ and $Z$, respectively. Both Alice and Bob send messages to Charlie, who is required to output joint random variables $X'Y'$ such that $\prob{XYX'Y'}{d\br{XY,X'Y'}\geq k}\leq\epsilon$, where $d:\supp{X}\times\supp{Y}\times\supp{X'}\times\supp{Y'}\rightarrow(0,+\infty)$ is a distortion measure. There is no shared randomness among the parties.  
\end{definition}

The following theorem obtains a nearly tight one-shot bound for this task, that uses auxiliary random variables of bounded size. Note that auxiliary random variables also arise in the characterization of the lossy source coding task (see \cite[Section 3.6]{GamalK12} for a discussion).

\begin{theorem}
\label{thm:lossyopt}
Given $\epsilon, \delta,\delta'\in(0,1)$ and joint distribution $XYZ$ over $\X\times\Y\times\Z$. For every $\br{k,\epsilon}$-lossy distributed source coding protocol with $R_1$ bits of communication from Alice to Charlie and $R_2$ bits of communication from Bob to Charlie, there exist random variables $MN$ taking values over a set $\M\times \N$ and satisfying $M-X-YNZ$, $N-Y-XMZ$ , $|\M|\leq |\X|, |\N|\leq |\Y|$. Moreover, there exists a function $f:\M\times\N\times \Z\rightarrow\X'\times\Y'$ satisfying  $\prob{XYZMN}{d\br{XY, f\br{M,N,Z}}\geq k}\leq\epsilon$ and
\[\prob{xyzmn\leftarrow XYZMN}{\frac{p_{M|XZ=x,z}\br{m}}{p_{M|NZ=n,z}\br{m}}\leq\frac{2^{R_1}}{\delta}\mbox{ and }\frac{p_{N|YZ=y,z}\br{n}}{p_{N|MZ=m,z}\br{n}}\leq\frac{2^{R_2}}{\delta} \atop \mbox{ and }\frac{p_{M|XZ=x,z}\br{m}p_{N|YZ=y,z}\br{n}}{p_{MN|Z=z}\br{m,n}}\leq\frac{2^{R_1+R_2}}{\delta}}\geq 1-3\delta.\]	
Furthermore, for any joint distribution $MN$ satisfying $M-X-YNZ$, $N-Y-XMZ$ and function $f:\M\times\N\times \Z\rightarrow\X'\times\Y'$ satisfying $\prob{XYZMN}{d\br{XY, f\br{M,N,Z}}\geq k}\leq\epsilon$, there exists a $\br{k,\epsilon+\delta+8\delta'}$-lossy distributed source protocol with communication $R_1$ from Alice to Charlie and communication $R_2$ from Bob to Charlie such that
\[\prob{xymn\leftarrow XYMN}{\frac{p_{M|XZ=x,z}\br{m}}{p_{M|NZ=n,z}\br{m}}\leq\delta' \cdot 2^{R_1}\mbox{ and }\frac{p_{N|YZ=y,z}\br{n}}{p_{N|MZ=m,z}\br{n}}\leq\delta'\cdot 2^{R_2},\atop\mbox{ and }~\frac{p_{M|XZ=x,z}\br{m}p_{N|YZ=y,z}\br{n}}{p_{MN|Z=z}\br{m,n}}\leq\delta'^4\cdot 2^{R_1+R_2-\log\log\frac{\max\set{\abs{\M},\abs{\N}}}{\delta'}}}\geq1-\delta.\]
\end{theorem}

\begin{proof}
We divide the proof in two parts.
\begin{itemize}
\item \textbf{Converse.} Given a $(k,\epsilon)$-lossy distributed source coding protocol, we can fix the local randomness used by Alice and Bob for encoding and obtain a new protocol where $X$ and $Y$ are mapped deterministically to the messages sent. Observe that this does not change the error and does not increase the communication cost. We choose $M$ and $N$ to be the messages sent by Alice and Bob respectively. Without loss of generality, we can assume that $|\M| \leq |\X|$ and $|\N|\leq |\Y|$. It holds that $M-X-YZN$ and $N-Y-XZM$.  Let $f$ be the function applied by Charlie to obtain $X'Y'$. By the correctness of the protocol, it holds that $\prob{XYZMN}{d\br{XY, f\br{M,N,Z}}\geq k}\leq\epsilon$. Let $U_1$ and $U_2$ be uniform distributions over $\M$ and $\N$, respectively. Note that for any $\br{x,y,z,m,n}$, 
	\begin{eqnarray}
		&&p_{M|XZ=x,z}\br{m}\leq 2^{R_1}p_{U_1}\br{m},\label{eqn:mxzu}\\
		&&p_{N|YZ=y,z}\br{n}\leq 2^{R_2}p_{U_2}\br{n},\label{eqn:nyzu}\\
		&&p_{MN|XYZ=x,y,z}\br{m,n}\leq 2^{R_1+R_2}p_{U_1}\br{m}p_{U_2}\br{n}.\label{eqn:mnxyzu}
	\end{eqnarray}
We further have
\begin{eqnarray}
	&&\prob{\br{m,n,z}\leftarrow MNZ}{p_{MN| Z=z}\br{m,n}\leq\delta p_{U_1}\br{m}p_{N|Z=z}\br{n}}\leq\delta,\label{eqn:mnzun}\\
	&&\prob{\br{m,n,z}\leftarrow MNZ}{p_{MN|Z=z}\br{m,n}\leq\delta p_{M|Z=z}\br{m}p_{U_2}\br{n}}\leq\delta,\label{eqn:mnzmu}\\
	&&\prob{\br{m,n,z}\leftarrow MNZ}{p_{MN|Z=z}\br{m,n}\leq\delta p_{U_1}\br{m}p_{U_2}\br{n}}\leq\delta.\label{eqn:mnzuu}
\end{eqnarray}
Then
\begin{eqnarray*}
	&&\prob{xyzmn\leftarrow XYZMN}{\frac{p_{M|XZ=x,z}\br{m}}{p_{M|NZ=n,z}\br{m}}\leq \frac{2^{R_1}}{\delta}}\\
	&&=\prob{xyzmn\leftarrow XYZMN}{p_{M|XZ=x,z}\br{m}p_{N|Z=z}\br{n}\leq \frac{2^{R_1}}{\delta}p_{MN|Z=z}\br{m,n}}\\
	&&\geq\prob{xymn\leftarrow XYMN}{\delta\cdot p_{U_1}\br{m}p_{N|Z=z}\br{n}\leq p_{MN|Z=z}\br{m,n}}\\
	&&\geq 1-\delta,
\end{eqnarray*}
where the first inequality is from Eq.~\eqref{eqn:mxzu} and the second inequality is from Eq.~\eqref{eqn:mnzun}.
Similarly, we have
\begin{eqnarray*}
	&&\prob{xyzmn\leftarrow XYZMN}{\frac{p_{N|YZ=y,z}\br{n}}{p_{N|MZ=m,z}\br{m}}\leq \frac{2^{R_2}}{\delta}}\geq1-\delta;\\
	&&\prob{xyzmn\leftarrow XYZMN}{\frac{p_{M|XZ=x,z}\br{m}p_{N|YZ=y,z}\br{n}}{p_{MN|Z=z}\br{m,n}}\leq \frac{2^{R_1+R_2}}{\delta}}\geq1-\delta.
\end{eqnarray*}
The converse follows from the union bound.

\item \textbf{Achievability.} Given joint distribution $XYZMN$ satisfying $M-X-YZN$ and $N-Y-XMZ$, the parties run the protocol present in Theorem~\ref{thm:main} with $\epsilon\leftarrow\delta,\delta\leftarrow\delta'$.  As guaranteed by Theorem~\ref{thm:main}, the protocol outputs random variables $M'N'$ satisfying that
\[\frac{1}{2}\onenorm{XYZMN-XYZM'N'}\leq \delta+8\delta'.\]
The Charlie applies $f$ to $\br{M',N'}$ to get $\br{X',Y'}$. It holds that
\begin{eqnarray*}
	&&\prob{xyx'y'\leftarrow XYX'Y'}{d\br{x,y,x',y'}\geq k}\\
	&&\leq\prob{xyzmn\leftarrow XYZMN}{d\br{x,y,f\br{m,n,z}}\geq k}+\frac{1}{2}\onenorm{XYZMN-XYZM'N'}\\
	&&\leq\epsilon+\delta+8\delta'
\end{eqnarray*}
This completes the proof.
\end{itemize}
\end{proof}

\subsection{Near optimal characterization of Task C in terms of the auxiliary random variables}

We first show how to reduce the amount of shared randomness in any given $(R_1, R_2, \epsilon)$ two-senders-one-receiver message compression with side information at the receiver, using an argument similar to Wyner~\cite{Wyner75com} and Newman~\cite{Newman91}. Since their arguments do not apply in the multi-partite setting (notice that the new randomness must be shared independently between Alice, Charlie and Bob, Charlie), we replace the Chernoff bound arguments in~\cite{Wyner75com, Newman91} with an argument based on bipartite convex-split lemma (Fact~\ref{genconvexcomb}). 
\begin{claim}
\label{clm:randreduct}
Fix $\epsilon, \delta\in (0,1)$. For any $\br{R_1, R_2, \epsilon}$ two-senders-one-receiver message compression with side information at the receiver, exists another $\br{R_1, R_2, \epsilon+2\delta}$ two-senders-one-receiver message compression with side information at the receiver that uses at most $\log\frac{24|\M||\N|}{\delta^3}$ bits of shared randomness between Alice and Charlie as well as Bob and Charlie. 
\end{claim}
\begin{proof}
Given a protocol where Alice sends $R_1$ bits to Charlie, Bob sends $R_2$ bits to Charlie and Charlie outputs $M'N'$ such that 
$\frac{1}{2}\|XYZMN - XYZM'N'\|_1\leq \epsilon,$ let $S_1$ be the shared randomness between Alice and Charlie and $S_2$ be the shared randomness between Bob and Charlie. Let $T_1$ be the message generate by Alice conditioned on $S_1, X$ and $T_2$ be the message generated by Bob conditioned on $S_2, Y$.  We apply Fact~\ref{genconvexcomb} to the random variables $XYZM'N'S_1S_2, S_1, S_2$ with $L_1, L_2$ chosen such that 
\begin{equation}
\label{eq:biconveq}
\prob{\stackrel{x,y,z,m,n,s_1,s_2\sim}{XYZM'N'S_1S_2}}{\frac{p_{XYZM'N'| S_1=s_1}(x,y,z,m,n)}{p_{XYZM'N'}(x,y,z,m,n)} \leq \frac{\delta^2}{24}\cdot L_1 \mbox{ and } \frac{p_{XYZM'N'| S_2=s_2}(x,y,z,m,n)}{p_{XYZM'N'}(x,y,z,m,n)} \leq \frac{\delta^2}{24}\cdot L_2 \atop\mbox{ and }\quad\frac{p_{XYZM'N'| S_1S_2=s_1,s_2}(x,y,z,m,n)}{p_{XYZM'N'}(x,y,z,m,n)} \leq \frac{\delta^2}{24}\cdot (L_1+L_2)} \geq 1-\delta,
\end{equation}
to obtain the random variable $XYZM'N'S_1^1\ldots S_1^{L_1}S_2^1\ldots S_2^{L_2}$ (with $S^i_1 = S_1$ and $S^j_2=S_2$) that satisfies
$$\frac{1}{2}\|XYZM'N'S_1^1\ldots S_1^{L_1}S_2^1\ldots S_2^{L_2} - XYZM'N'\times S_1\times\ldots S_1\times S_2\times \ldots S_2\|_1\leq 2\delta.$$
This expression can be reaaranged to obtain
$$\expec{\stackrel{s_1^1,\ldots s_1^{L_1},s_2^1,\ldots s_2^{L_2}\sim}{S_1\times\ldots S_1\times S_2\times \ldots S_2}}{\frac{1}{2}\onenorm{\frac{1}{L_1L_2}\sum_{(i,j)\in [L_1]\times[L_2]}XYZ(M'N'\mid S_1S_2=s^i_1,s^j_2)  - XYZM'N'}}\leq 2\delta.$$
Thus, there exists a choice of $\{s_1^1,\ldots s_1^{L_1},s_2^1,\ldots s_2^{L_2}\}$ such that
\begin{equation}
\label{eq:averageclose}
\frac{1}{2}\onenorm{\frac{1}{L_1L_2}\sum_{(i,j)\in [L_1]\times[L_2]}XYZ(M'N'\mid S_1S_2=s^i_1,s^j_2)  - XYZM'N'} \leq 2\delta.
\end{equation}
The new protocol is as follows.
\begin{itemize}
\item Alice and Charlie share uniform randomness $U_1$ taking values in $[L_1]$. Bob and Charlie share uniform randomness $U_2$ taking values in $[L_2]$. 
\item Conditioned on the value $i\sim U_1$, Alice generates $(T_1| XYZ, S_1=s^i_1)$ and sends it to Charloe. Conditioned on the value $j\sim U_2$, Bob generates $(T_2| XYZ, S_2= s^j_2)$ and sends it to Charlie. 
\item Charlie, who also observes $(i,j)\sim U_1\times U_2$, generates $M'N'$ conditioned on $(s_1^i, s_2^j)$. 
\item Let the output of Charlie, averaged over the shared randomness, be $M''N''$. 
\end{itemize}
It holds that 
$$XYZM''N'' = \frac{1}{L_1L_2}\sum_{(i,j)\in [L_1]\times[L_2]}XYZ(M'N'\mid S_1S_2=s^i_1,s^j_2).$$
Thus, Eq.~\eqref{eq:averageclose} guarantees that
$$\frac{1}{2}\|XYZM''N'' - XYZMN\|_1 \leq \epsilon+2\delta.$$
To bound the size of shared randomness, observe that Eq.~\eqref{eq:biconveq} can be rephrased as follows, using the fact that $XYZ$ is independent of $S_1S_2$:
\begin{equation*}
\prob{\stackrel{x,y,z,m,n,s_1,s_2\sim}{XYZM'N'S_1S_2}}{\frac{p_{M'N'|XYZS_1=x,y,z,s_1}(m,n)}{p_{M'N'| XYZ= x,y,z}(m,n)} \leq \frac{\delta^2}{24}\cdot L_1 \mbox{ and } \frac{p_{M'N'| XYZS_2=x,y,z,s_2}(m,n)}{p_{M'N'|XYZ= x,y,z}(m,n)} \leq \frac{\delta^2}{24}\cdot L_2 \atop\mbox{ and }\quad\frac{p_{M'N'| XYZS_1S_2=x,y,z,s_1,s_2}(m,n)}{p_{M'N'| XYZ=x,y,z}(m,n)} \leq \frac{\delta^2}{24}\cdot (L_1+L_2)} \geq 1-\delta.
\end{equation*}
Since 
$$\prob{x,y,z,m,n \sim XYZMN}{p_{M'N'|XYZ= x,y,z}(m,n) \leq \frac{\delta}{|\M||\N|}} \leq \delta,$$
and probabilities are less than $1$, we have 
\begin{eqnarray*}
&&\prob{\stackrel{x,y,z,m,n,s_1,s_2\sim}{XYZM'N'S_1S_2}}{\frac{p_{M'N'|XYZS_1=x,y,z,s_1}(m,n)}{p_{M'N'| XYZ= x,y,z}(m,n)} \leq \frac{\delta^2}{24}\cdot L_1 \mbox{ and } \frac{p_{M'N'| XYZS_2=x,y,z,s_2}(m,n)}{p_{M'N'|XYZ= x,y,z}(m,n)} \leq \frac{\delta^2}{24}\cdot L_2 \atop\mbox{ and }\quad\frac{p_{M'N'| XYZS_1S_2=x,y,z,s_1,s_2}(m,n)}{p_{M'N'| XYZ=x,y,z}(m,n)} \leq \frac{\delta^2}{24}\cdot (L_1+L_2)}\\
&&\geq \prob{\stackrel{x,y,z,m,n,s_1,s_2\sim}{XYZM'N'S_1S_2}}{\frac{1}{p_{M'N'| XYZ= x,y,z}(m,n)} \leq \frac{\delta^2}{24}\cdot L_1 \mbox{ and } \frac{1}{p_{M'N'|XYZ= x,y,z}(m,n)} \leq \frac{\delta^2}{24}\cdot L_2 \atop\mbox{ and }\quad\frac{1}{p_{M'N'| XYZ=x,y,z}(m,n)} \leq \frac{\delta^2}{24}\cdot (L_1+L_2)}\\
&&\geq \prob{\stackrel{x,y,z,m,n\sim}{XYZM'N'}}{\frac{1}{p_{M'N'| XYZ= x,y,z}(m,n)} \leq \frac{\delta^2}{24}\cdot 2^{\min\br{L_1, L_2}}} \geq 1-\delta,
\end{eqnarray*}
if we choose $L_1 = \frac{24|\M||\N|}{\delta^3}$ and $L_2 = \frac{24|\M||\N|}{\delta^3}$. This completes the proof.
\end{proof}

We have the following theorem for Task 1 (Definition~\ref{def:task1}), in terms of auxiliary random variables of bounded size. The proof closely follows the proof of Theorem~\ref{thm:lossyopt} and uses above claim.
\begin{theorem}
\label{thm:auxchar}
Fix $\epsilon, \delta\in (0,1)$. For any $\br{R_1, R_2, \epsilon}$ two-senders-one-receiver message compression with side information, there exist random variables $S_1S_2T_1T_2$ jointly correlated with $XYZ$ such that $S_1S_2XYZ= S_1\times S_2\times XYZ$, $T_1-S_1X-YZS_2T_2$, $T_2-S_2Y-XZS_1T_1$ and a function $f: \Z\times\supp{S_1}\times \supp{S_2}\times \supp{T_1}\times \supp{T_2}\rightarrow \Z\times\M\times \N$ such that $$\frac{1}{2}\|XYf(ZS_1T_1S_2T_2) - XYZMN\|_1\leq \epsilon + 2\delta,$$ $$|S_1| \leq  \frac{24|\M||\N|}{\delta^3},\quad |S_2|\leq  \frac{24|\M||\N|}{\delta^3},\quad |T_1|\leq |\X||S_1|, \quad |T_2|\leq |\Y||S_2|$$ and
\[\prob{xyzt_1s_1t_2s_2\leftarrow XYZT_1S_1T_2S_2}{\frac{p_{T_1S_1|XZ=x,z}\br{t_1,s_1}}{p_{T_1S_1|T_2S_2Z=t_2,s_2,z}\br{t_1,s_1}}\leq\frac{2^{R_1}}{\delta}\mbox{ and }\frac{p_{T_2S_2|YZ=y,z}\br{t_2,s_2}}{p_{T_2S_2|T_1S_1Z=t_1,s_1,z}\br{t_2,s_2}}\leq\frac{2^{R_2}}{\delta} \atop \mbox{ and }\frac{p_{T_1S_1|XZ=x,z}\br{t_1,s_1}p_{T_2S_2|YZ=y,z}\br{t_2,s_2}}{p_{T_1S_1T_2S_2|Z=z}\br{t_1,s_1,t_2,s_2}}\leq\frac{2^{R_1+R_2}}{\delta}}\geq 1-3\delta.\]
Furthermore, for every random variables $S_1S_2T_1T_2$ jointly correlated with $XYZ$ such that $S_1S_2XYZ= S_1\times S_2\times XYZ$, $T_1-S_1X-YZS_2T_2$, $T_2-S_2Y-XZS_1T_1$ and a function $f: \Z\times\supp{S_1}\times \supp{S_2}\times \supp{T_1}\times \supp{T_2}\rightarrow \Z\times\M\times \N$ such that $\frac{1}{2}\|XYf(ZS_1T_1S_2T_2) - XYZMN\|_1\leq \epsilon$, there exists a $(R_1, R_2, \epsilon+9\delta)$ two-senders-one-receiver message compression with side information for any $R_1, R_2$ satisfying
\[\prob{xyzt_1s_1t_2s_2\leftarrow XYZT_1S_1T_2S_2}{\frac{p_{T_1S_1|XZ=x,z}\br{t_1,s_1}}{p_{T_1S_1|T_2S_2Z=t_2,s_2,z}\br{t_1,s_1}}\leq\delta 2^{R_1}\mbox{ and }\frac{p_{T_2S_2|YZ=y,z}\br{t_2,s_2}}{p_{T_2S_2|T_1S_1Z=t_1,s_1,z}\br{t_2,s_2}}\leq\delta 2^{R_2} \atop \mbox{ and }\frac{p_{T_1S_1|XZ=x,z}\br{t_1,s_1}p_{T_2S_2|YZ=y,z}\br{t_2,s_2}}{p_{T_1S_1T_2S_2|Z=z}\br{t_1,s_1,t_2,s_2}}\leq\delta^4 2^{R_1+R_2}}\geq 1-\delta.\]
\end{theorem}
\begin{proof}The proof is divided in two parts.
\begin{itemize}
\item {\bf Converse:} Given a $\br{R_1, R_2, \epsilon}$ protocol, we can assume without loss of generality that all the randomness Alice adapts is shared between Alice and Charlie and all the randomness Bob adapts is shared between Bob and Charlie. From this, we use Claim~\ref{eq:biconveq} to construct a $\br{R_1, R_2, \epsilon+2\delta}$ protocol that uses $\frac{24|\M||\N|}{\delta^3}$ bits of shared randomness between Alice and Charlie, and $\frac{24|\M||\N|}{\delta^3}$ bits of shared randomness between Bob and Charlie. We denote by $S_1$ the shared randomness between Alice and Charlie, and by $S_2$ the shared randomness between Bob and Charlie. Let $T_1$ be the message sent from Alice to Charlie and $T_2$ be the message sent from Bob to Charlie. Since $T_1$ is obtained by applying deterministic function on $XS_1$, we have that $|T_1|\leq |\X||S_1|$. Similarly, $|T_2|\leq |\Y||S_2|$. Let $U_1$ be uniform over $\supp{T_1}$ and $U_2$ be uniform over $\supp{T_2}$. Let $f$ be the function that Charlie applies on $S_1,S_2, T_1, T_2,Z$ to obtain $M'N'$. Then 
\begin{eqnarray*}
		&&p_{T_1S_1|XZ=x,z}\br{t_1,s_1}\leq 2^{R_1}p_{U_1}\br{t_1}p_{S_1}\br{s_1}\\
		&&p_{T_2S_2|YZ=y,z}\br{t_2,s_2}\leq 2^{R_2}p_{U_2}\br{t_2}p_{S_2}\br{s_2}\\
		&&p_{S_1S_2T_1T_2|XYZ=x,y,z}\br{t_1,s_1, t_2, s_2}\leq 2^{R_1+R_2}p_{U_1}\br{t_1}p_{U_2}\br{t_2} p_{S_1S_2}\br{s_1,s_2}.
\end{eqnarray*}
The rest of the proof follows closely the converse proof given in Theorem~\ref{thm:lossyopt}. 
\item {\bf Achievability:} The achievability also follows along the lines similar to Theorem~\ref{thm:lossyopt}. By a straightforward application of Theorem~\ref{thm:main}, Alice and Bob communicate $R_1$ and $R_2$ bits respectively to Charlie such that Charlie is able to output $S'_1S'_2T'_1T'_2$ satisfying
$$\frac{1}{2}\|XYZS'_1S'_2T'_1T'_2 - XYZS_1S_2T_1T_2\|_1\leq 9\delta.$$
Charlie now applies the function $f$ to obtain the desired output. It holds that
$$\frac{1}{2}\|XYf(ZS'_1S'_2T'_1T'_2) - XYZMN\|_1\leq \epsilon+9\delta.$$
This completes the proof.
\end{itemize} 
\end{proof}

\subsection{Recovering achievable communication for DSC task and Task B}
\label{subsec:recoverDSC}

Another application is the following corollary, for the problem of DSC. While it is a special case of lossy distributed source coding, it is possible to obtain a simpler bound without introducing auxiliary random variables. We reproduce the near-optimal one-shot bound given in~\cite{AnshuJW17un}, up to an additive factor of $\log\log\max\{\abs{\X},\abs{\Y}\}$. 

\begin{cor}\label{cor:DSC}
Let $\epsilon\in(0,1)$ such that $\frac{1}{\sqrt{\delta}}$ is an integer. Let $R_1,R_2$ satisfy
\begin{equation}
\prob{x,y\sim XY}{\frac{1}{p_{X|Y=y}\br{x}}\leq\delta\cdot 2^{R_1}~\mbox{and}~\frac{1}{p_{Y|X=x}\br{y}}\leq\delta\cdot 2^{R_2}\atop\mbox{and}~\frac{1}{p_{XY}\br{xy}}\leq\delta^4\cdot2^{R_1+R_2-\log\log\frac{\max\set{\abs{\X},\abs{\Y}}}{\delta}}}\geq1-\epsilon.
\end{equation}
There exists a protocol satisfies the following:
\begin{itemize}
\item No players share public coins.

\item Alice and Bob observe a sample from $X$ and $Y$ and then send $R_1+3\log\frac{1}{\delta}$ bits and $R_2+3\log\frac{1}{\delta}$ bits to Charlie, respectively;

\item  Charlie outputs the random variables $X'Y'$ such that 
\[\Pr\{XY\neq X'Y'\}\leq \epsilon+8\delta.\]
\end{itemize}

\end{cor}

\begin{proof}
Applying Theorem~\ref{thm:main} with $Z$ trivial, $M=X$ and $N=Y$, we obtain a randomness assisted protocol with communications $R_1$ and $R_2$ from Alice and Bob respectively. Charlie outputs random variables $X',Y'$ such that 
$$\frac{1}{2}\|XYXY - XYX'Y'\|_1\leq \epsilon+8\delta.$$ Let $\S= \{(x,y,x,y): x\in \X, y\in \Y\}$. Then 
$$1- \Pr\{XY= X'Y'\}=\left|\Pr_{XYXY}\{\S\} - \Pr_{XYX'Y'}\{\S\}\right|\leq \frac{1}{2}\|XYXY - XYX'Y'\|_1\leq \epsilon+8\delta.$$ This completes the proof by the standard derandomization argument to fix the shared randomness.
\end{proof}

Using the argument similar to Theorem \ref{maintask1}, we have the following corollary which expresses above communication region in terms of conditional entropies.
\begin{cor}\label{cor:DSC2}
Let $\epsilon\in(0,1)$ such that $\frac{1}{\sqrt{\delta}}$ is an integer. Let $R_1,R_2$ satisfy:
	\begin{enumerate}
	\item $R_1\geq \frac{3H\br{X|Y}}{\epsilon}+\log\frac{1}{\delta}$,	
	\item $R_2\geq\frac{3H(Y|X)}{\epsilon}+\log\frac{1}{\delta}$,\item $R_1+R_2\geq\frac{3H\br{XY}}{\epsilon}+4\log\frac{1}{\delta}+\log\log\frac{\max\set{\abs{\X},\abs{\Y}}}{\delta}$
	\end{enumerate}
	There exists a protocol satisfies the following:
	\begin{itemize}
		\item No players share public coins.
		
		\item Alice and Bob observe a sample from $X$ and $Y$ and then send $R_1+3\log\frac{1}{\delta}$ bits and $R_2+3\log\frac{1}{\delta}$ bits to Charlie, respectively;	
		\item  Charlie outputs the random variables $X'Y'$ such that 
		\[\Pr\{XY\neq X'Y'\}\leq \epsilon+8\delta.\]
	\end{itemize}
\end{cor}

We also reproduce the main results in~\cite{BravermanRao11} and~\cite{AnshuJW17un} for Task B, up to additive factor of $\log\log|\M|$. This is obtained by setting $Y$ and $N$ to be trivial in Theorem~\ref{thm:main}.

\begin{cor}
\label{cor:1sender1receiver}
Let $\epsilon\in(0,1)$ such that $\frac{1}{\sqrt{\delta}}$ is an integer. Let $R$ satisfy
\begin{equation}
\prob{x,m,z\sim XMZ}{\frac{p_{M|X=x}\br{m}}{p_{M\mid Z=z}\br{m}}\leq\delta^4\cdot2^{R-\log\log\frac{\abs{\M}}{\delta}}}\geq1-\epsilon.
\end{equation}
There exists a protocol satisfies the following:
\begin{itemize}
\item Alice and Charlie share public random coins;

\item Alice sends $R+3\log\frac{1}{\delta}$ bits to Charlie;

\item  Charlie outputs the random variables $M'$ such that 
\[\frac{1}{2}\onenorm{XZM-XZM'}\leq \epsilon+8\delta.\]
\end{itemize}
\end{cor}

\section{Lower bound on the expected communication cost of one-way protocols for Slepian-Wolf task}
\label{sec:expeccomm}

In this section, we consider the expected communication cost of Slepian-Wolf task, which is a special case of Task B where $M=X$ (originally studied by Slepian and Wolf~\cite{SlepianW73}). The protocol in~\cite{BravermanRao11} for Task B implies that there is an interactive protocol achieving the expected communication $H(X|Z) + c\br{\sqrt{H(X|Z)} + \log\frac{1}{\epsilon}}$, for some constant $c$ independent of $\abs{\X}, \abs{\Z}$. The following theorem shows that interaction is necessary, giving a much larger lower bound for one-way protocols.

\begin{theorem}
\label{theo:expeclowbound}
	For any integer $N$ and $\epsilon\in(0,\frac{1}{64})$, there exists a joint distribution $XZ$ with support $[\br{1-\sqrt{\epsilon}}N]\times[N]$ such that the expected communication cost of any one-way protocol achieving Task B with $M=X$ and error at most $\epsilon$ is at least
	\[\frac{1}{6\sqrt{\epsilon}}H\br{X|Z}.\]
\end{theorem}

\begin{proof}
	Set $\delta\defeq\sqrt{\epsilon}$. Let $Z$ be a uniform distribution over $[N]$. For any $z\in[\delta N]$, set $\br{X|Z=z}$ to be the uniform distribution over $[\br{1-\delta}N]$. For any $z\in\set{\delta N+1,\ldots, N}$, set $\br{X|Z=z}=\id\br{X=z-\delta N}$. Then $p_X\br{x}=\delta\frac{1}{\br{1-\delta}N}+\frac{1}{N}=\frac{1}{\br{1-\delta}N}$. Thus $X$ is uniform over $[\br{1-\delta}N]$. Furthermore, $H\br{X|Z}=\sum_{z=1}^{\delta N}\frac{1}{N}H\br{X}=\delta\log\br{1-\delta}N\leq\delta\log N$.
	
	Given a protocol with expected communication cost $C$ and expected error $\epsilon$, without loss of generality, we may assume all the randomness occurring in the protocol are shared between Alice and Charlie, which is denoted by $R$. Conditioning on inputs $\br{x,z}$ and randomness $r$, let $\ell_{x,r}$ be the length of the message Alice sends to Charlie. Charlie outputs $x'\br{x,z,r}$, which is a function of $x,z,r$. Define $\epsilon_{x,z,r}\defeq\id\br{x'\br{x,z,r}=x}$. We have
	\[\sum_{x,r}p_X\br{x}p_R\br{r}\ell_{x,r}=C, \quad \sum_{x,z,r}p_{XZ}\br{x,z}p_R\br{r}\epsilon_{x,z,r}=\epsilon.\]
	By Markov inequality, there exists a $r_0$ such that 
	
	\[\sum_xp_X\br{x}\ell_{x,r_0}\leq 3C,~\sum_{x,z}p_{XZ}\br{x,z}\epsilon_{x,z,r_0}\leq 3\epsilon.\] 
	Then
	\[3\epsilon\geq\sum_zp_Z\br{z}\sum_xp_{X|Z=z}\br{x}\epsilon_{x,z,r_0}\geq\frac{1}{N}\sum_{z=1}^{\delta N}\sum_xp_{X|Z=z}\br{x}\epsilon_{x,z,r_0}.\]
	Thus, there exists $z_0\in[\delta N]$ such that $\sum_xp_{X|Z=z_0}(x)\epsilon_{x,z_0,r_0}\leq\frac{3\epsilon}{\delta}$. Note that $p_{X|Z=z}\br{x}=p_X\br{x}$ for $z\in[\delta N]$. We conclude that 
	\[\sum_xp_X\br{x}\epsilon_{x,z_0,r_0}\leq\frac{3\epsilon}{\delta}= 3\sqrt{\epsilon},~\sum_xp_X\br{x}\ell_{x,r_0}\leq 3C.\]
	
	We consider a new protocol where Alice receives $x\sim X$ and sends the message determined by $\br{x,r_0}$ to Charlie, which is of length $\ell_{x,r_0}$. Charlie outputs $x'(x, z_0, r_0)$ deterministically according to the message, as in the above protocol. As $\epsilon_{x,z,r}$ is either $0$ or $1$, the protocol makes zero error on at least $\br{1-3\sqrt{\epsilon}}$ fraction of $x$. Thus
	\[3C\geq\log\br{\br{1-3\sqrt{\epsilon}}\br{1-\delta}N}\geq\frac{1}{2}\log N\geq\frac{1}{2\sqrt{\epsilon}}H\br{X|Z}.\]
	This completes the proof.

\end{proof}

As a consequence, there is no one-way protocol for Task C that can achieve the communication region in Eq.~\eqref{eq:timesharetaskB} in expected communication, even when the register $Z$ is trivial. Now, we sketch a simple one-shot interactive protocol that achieves the communication region in Eq.~\eqref{eq:timesharetaskB}, where there is no interaction between Alice and Bob. Observe that the corner points of the region in Eq.~\eqref{eq:timesharetaskB} are 
$$(\condmutinf{X}{M}{NZ}, \condmutinf{Y}{N}{Z}) \mbox{   and   } (\condmutinf{X}{M}{Z},\condmutinf{Y}{N}{MZ}).$$ 
To achieve the point 
$$p(\condmutinf{X}{M}{NZ}, \condmutinf{Y}{N}{Z}) + (1-p) (\condmutinf{X}{M}{Z},\condmutinf{Y}{N}{MZ}),$$
Charlie prepares a random variable $B_1B_2B_3$ such that
\begin{eqnarray*}     
&&p_{B_1B_2B_3}\br{b_1,b_2,b_3}=
\begin{cases}
p~&\mbox{if $b_1=b_2=b_3=0$}
\\
1-p~&\mbox{if $b_1=b_2=b_3=1$}
\end{cases}
\end{eqnarray*}
He sends $B_1$ to Alice and $B_2$ to Bob. Conditioned on the value $0$, Bob and Charlie run the protocol in~\cite{BravermanRao11} 
to communicate $N$ to Charlie and then Alice and Charlie run the protocol in~\cite{BravermanRao11} to communicate $M$ to Charlie. Conditioned on the value $1$, first Alice and Charlie run the protocol to communicate $M$ and then Bob and Charlie run the protocol to communicate $N$. The expected communication cost is $p$ times the expected communication cost of the former protocol and $1-p$ times the expected communication cost of the latter. This achieves the desired result.

\subsection*{Conclusion}

In this work, we have studied the problem of message compression in the multi-party setting. We have obtained an achievable communication region that can be viewed as a one-shot analogue of the time sharing region for Task C. Since time-sharing is not possible in the one-shot setting, we have developed a novel hypothesis testing approach to obtain our main result. As applications of our result, we obtain near optimal one-shot communication regions for Task C and the lossy distributed source coding task, in terms of auxiliary random variables. A utility of our result is that the auxiliary variables involved are of size comparable to the size of random variables input to the task. This feature is often useful from the computational point of view and present in the characterization of communication for various tasks (see \cite{GamalK12} for such examples). We leave open the problem of obtaining a near optimal characterization without using auxiliary random variables, which is not known also for the task of source coding with a helper. An important question that we do not answer is about formulating a proper notion of information complexity \cite{Braverman15} in the interactive setting. We believe our compression results will shed light on this, as the notion of information complexity is closely tied to compression protocols in the two-party setting \cite{HJMR10, BravermanRao11, BarakBCR13}. 

\subsection*{Acknowledgment}

We thank Rahul Jain and Naqueeb Ahmad Warsi for helpful discussions, and especially thank Rahul Jain for pointing out the connection between Wyner's result \cite{Wyner75com} and the convex-split method. A.A. is supported by the National Research Foundation, Prime Minister's Office, Singapore and the Ministry of Education, Singapore under the Research Centres of Excellence programme. Part of the work was done when A.A. was visiting State Key Laboratory for Novel Software Technology, Nanjing University sponsored by the National Key R$\&$ D Program of China 2018YFB1003202. P. Y. is supported by the National Key R$\&$ D Program of China 2018YFB1003202 and a China Youth 1000-Talent grant. 

\bibliographystyle{ieeetr}
\bibliography{References}

\appendix

\section{Task C and the SCH task}
\label{append:SCHtaskB}

It is not immediately clear if the task of source coding with a helper is equivalent to Task C with $M=X$ and $N$ is trivial. This is because in the former, the only requirement is that Charlie outputs the correct $X$ with high probability (averaged over $Y$), whereas in the latter it is required that the global distribution is obtained with small error in $\ell_1$- distance. We show here that both definitions are equivalent up to constant factor increase in error and hence Task C with $M=X$ and $N$ trivial is equivalent to the task of source coding with a helper. More precisely, we have the following claim.

\begin{claim}
Fix $\epsilon \in (0,1)$. Let $YXX'$ be joint random variables such that $\Pr\{X\neq X'\}\leq \epsilon$. Then it holds that
$$\frac{1}{2}\|YXX'-YXX\|_1\leq 4\epsilon,$$
where the random variable $YXX$ is defined as 
\begin{eqnarray*}
p_{YXX}(y,x,x') =
\begin{cases}
p_{YX}(y,x)~&\mbox{if $x=x'$}
\\
0~&\mbox{otherwise}
\end{cases}
\end{eqnarray*}
\end{claim}
\begin{proof}
Let $X_1Y_1X'_1$ be the random variable obtained by the restriction of $XYX'$ to the set $\{(x,y,x)\}$. By Fact~\ref{fac:distancerestriction}, we have
\begin{equation}
\label{xyx'close}
\frac{1}{2}\|Y_1X_1X'_1 - YXX'\|_1 \leq 1- \Pr\{X=X'\} \leq \epsilon.
\end{equation}
By Fact~\ref{lem:distancemonotone}, this implies
\begin{equation}
\label{eq:x1xclose}
\frac{1}{2}\|X_1X'_1 - XX\|_1\leq \epsilon
\end{equation}
and 
\begin{equation}
\label{eq:x1y1close}
\frac{1}{2}\|X_1Y_1 - XY\|_1\leq \epsilon \implies \frac{1}{2}\sum_xp_X(x)\|(Y_1 \mid X_1=x) - (Y\mid X=x)\| \leq 2\epsilon.
\end{equation}
Define the random variable $Y_1XX$ as follows:
\begin{eqnarray*}
p_{Y_1XX}(y,x,x') = p_{Y\mid X_1=x}(y)\cdot p_{XX}(x,x).
\end{eqnarray*}
Then Eq.~\eqref{eq:x1xclose} implies that 
$$\frac{1}{2}\|Y_1X_1X'_1 - Y_1XX\| \leq \epsilon.$$ On the other hand, Eq.~\eqref{eq:x1y1close} implies that
$$\frac{1}{2}\|Y_1XX - YXX\|_1 \leq 2\epsilon.$$ Combining, we conclude 
$$\frac{1}{2}\|Y_1X_1X'_1 - YXX\|_1 \leq 3\epsilon.$$ Using this with Eq.~\eqref{xyx'close}, we conclude that
$$\frac{1}{2}\|YXX' - YXX\|_1 \leq 4\epsilon,$$ which completes the proof. 
\end{proof}

\section{Proof of Claim \ref{clm:counterexample}}
\label{append:counterex}

Fix $\alpha\in (0,1)$  and set $\epsilon = \frac{\alpha(1-\alpha)^2}{2}$.  Assume $\X\defeq [\abs{\X}]$ and choose $X=Y$. Let 
\begin{eqnarray*}     
p_{X}\br{x}\defeq
\begin{cases}
\alpha &\mbox{if $x=1$}
\\
\frac{1-\alpha}{|\X|-1}~&\mbox{otherwise,}\end{cases}
\end{eqnarray*}
Let $\M=\N=\X$ and consider 
\begin{eqnarray*}     
p_{M\mid X=x}\br{m}\defeq
\begin{cases}
\alpha &\mbox{if $m=x$}
\\
\frac{1-\alpha}{|\X|-1}~&\mbox{otherwise,}\end{cases}
\end{eqnarray*}
\begin{eqnarray*}     
p_{N\mid X=x}\br{n}\defeq
\begin{cases}
\alpha &\mbox{if $n=x$}
\\
\frac{1-\alpha}{|\X|-1}~&\mbox{otherwise,}\end{cases}
\end{eqnarray*}
Then 
\begin{eqnarray*}     
p_{MN}\br{m,n}=
\begin{cases}
\alpha^3 + \frac{(1-\alpha)^3}{(|\X|-1)^3} &\mbox{if $m=n=1$}
\\
\frac{\alpha^2(1-\alpha)}{|\X|-1}+ \frac{\alpha(1-\alpha)^2}{(|\X|-1)^2}+ \frac{(1-\alpha)^3(|\X|-2)}{(|\X|-1)^3} &\mbox{if $n\neq m=1$}
\\
\frac{\alpha^2(1-\alpha)}{|\X|-1}+ \frac{\alpha(1-\alpha)^2}{(|\X|-1)^2}+ \frac{(1-\alpha)^3(|\X|-2)}{(|\X|-1)^3} &\mbox{if $m\neq n=1$}
\\
\frac{\alpha^2(1-\alpha)}{|\X|-1}+ \frac{\alpha(1-\alpha)^2}{(|\X|-1)^2}+ \frac{(1-\alpha)^3(|\X|-2)}{(|\X|-1)^3} &\mbox{if $m= n\neq 1$}
\\
\frac{3\alpha(1-\alpha)^2}{(|\X|-1)^2} + \frac{(|\X|-3)(1-\alpha)^3}{(|\X|-1)^3}~&\mbox{otherwise,}\end{cases}
\end{eqnarray*}
The probability mass of the set $\set{(x,x,n): x\neq 1, n\neq x, n\neq 1\}\subset \X \times \M\times \N}$ under the distribution $p_{XMN}$ is $\alpha(1-\alpha)^2 \geq \epsilon$. Hence, there exists a tuple $(x,x,n) \in \X\times \M\times \N$ with $x\neq 1, n\neq x, n\neq 1$ such that   
$$c' \geq \log\br{\frac{p_{M\mid X=x}(x)p_{N\mid X=x}(n)}{p_{MN}(x,n)}} \geq \log|\X| - \log(1+\alpha-2\alpha^2)- \frac{c_1(\alpha)}{|\X|},$$
for some $c_1(\alpha)$ that only depends on $\alpha$. On the other hand,
\begin{eqnarray*}
H(MN) &\leq& \alpha^3\log\frac{1}{\alpha^3} + 3\alpha^2(1-\alpha)\log\frac{|\X|}{3\alpha^2(1-\alpha)} + (1-\alpha)^2(1+2\alpha)\log\frac{|\X|}{ (1-\alpha)^2(1+2\alpha)}  + \frac{c_2(\alpha)}{|\X|}\\
&=& (1-\alpha^3)\log|\X| + \mathrm{H}\br{\{\alpha^3, 3\alpha^2(1-\alpha), (1-\alpha)^2(1+2\alpha)\}}+ \frac{c_2(\alpha)}{|\X|},
\end{eqnarray*}
for some $c_2(\alpha)$ that only depends on $\alpha$. Thus, for large $|\X|$, we find
$$\frac{H(MN)}{c'} \leq 1-\alpha^3 + \frac{c_3(\alpha)}{|\X|},$$
for some $c_3(\alpha)$ that only depends on $\alpha$. We can solve $\alpha$ in terms of $\epsilon$ to conclude that either $\alpha \leq 4\epsilon$ or $1-\alpha \leq \sqrt{4\epsilon}$.
Using the second bound, we obtain that
$$\frac{H(MN)}{c'} \leq 6\sqrt{\epsilon} + \frac{c_3(\alpha)}{|\X|} \leq 7\sqrt{\epsilon},$$
giving the desired upper bound
\end{document}